\documentclass[12pt]{article}

\usepackage{amssymb,amsmath,amsfonts,amsthm}
\usepackage{graphicx}
\usepackage{epstopdf}
\usepackage{caption}
\usepackage[dvipsnames,usenames]{color}
\usepackage{float}
\usepackage{subfigure}
\usepackage[toc,page]{appendix}

\newtheorem{theorem}{Theorem}

\newtheorem{remark}{Remark}
\newtheorem{definition}{Definition}
\usepackage{hyperref}
\usepackage{authblk}

\input{epsf.tex}

\begin{document}

\numberwithin{equation}{section}
\numberwithin{theorem}{section}

\title{A Random Access G-Network: Stability, Stable Throughput, and Queueing Analysis}
\author[*]{Ioannis Dimitriou}
\author[**]{Nikolaos Pappas}
\affil[*]{Department of Mathematics, University of Patras, 26500, Patras, Greece}
\affil[**]{Department of Science and Technology, Link\"{o}ping University, Campus Norrk\"{o}ping, 60174, Sweden}
\affil[ ]{\textit {idimit@math.upatras.gr, nikolaos.pappas@liu.se}}
\renewcommand\Authands{ and }
\providecommand{\keywords}[1]{\textbf{\textit{Keywords---}} #1}
\maketitle

\begin{abstract}
The effect of signals on stability, throughput region, and delay in a two-user slotted ALOHA based random-access system with collisions is
considered. This work gives rise to the development of random access G-networks, which can model virus attacks or other malfunctions and introduce load balancing in highly interacting networks. The users are equipped with infinite
capacity buffers accepting external bursty arrivals. We consider both negative and triggering signals. Negative signals delete a packet from a user queue, while triggering signals cause the instantaneous transfer of packets among user queues. We obtain the exact stability region, and show that the stable throughput region is a subset of it. Moreover, we perform a compact mathematical analysis to obtain exact expressions for the queueing delay by solving a Riemann boundary value problem. A computationally efficient way to obtain explicit bounds for the queueing delay is also presented. The theoretical findings are numerically evaluated and insights regarding the system performance are derived.  
\end{abstract}


\section{Introduction}
This work focus on the modeling and analysis of a two-user buffered ALOHA type multiple access system by taking into account the presence of signals, which either delete a packet or trigger the instantaneous movement of a packet between the user buffers. More precisely, this work introduce the so-called called Random Access G-network (RAG-Network) to model virus attacks, as well as probabilistic load balancing in highly interacting multiple access communication. 

The exact modelling of the system we consider, is in terms of a two dimensional random walk in the quarter plane, which posses a partial spatial homogeneity, due to the high level of interaction among the queues. The investigation of interacting queueing system is a challenging task with severe mathematical difficulties, which have been initially reported in the seminal work in \cite{fay2}, and further in \cite{fay1,coh1}. They have also received considerable attention due to their applicability in random access networks \cite{nain,szpa1,PappasJCN2016, dimpaptwc,dimpap,adhoc} and in shared processor systems \cite{guillemin2004,dim1,dim3,borst}. On the other hand, G-networks (or queueing networks with negative customers, signals, triggers, etc.) introduced in \cite{GN1}, are characterized by the following feature: in addition to the conventional customers, negative customers or signals arriving to a non-empty queue remove an amount of work from the queue or transfer it to another queue \cite{GN4}. The analysis of this versatile class of networks has significantly enriched queueing theory as well as contributed to the development of real applications in fields such as computers, communications, manufacturing, energy as well as self-aware networks.

The goal of this work is to provide a general framework in introducing and analysing multiple access systems with the feature of negative and triggering signals. Our ultimate goal is to provide a general model to describe virus attacks or other malfunctions, and load balancing in multiple access networks as well as to study their impact on the overall system performance. 
\subsection{Related Work}
\paragraph{G-Networks:}
Motivated by neural
network modelling \cite{gelnn}, a novel stochastic network, called G-network or queueing network with signals was introduced as a unifying model for neural and queuing networks. In contrast to traditional queueing networks where (positive) customers obey the specified service and routing disciplines determined by the network dynamics, there is another type of customers with the effect of signal that interact upon arrival at a queue with the queue or with the backlogged customers. 

G-Networks \cite{GN1} establish a versatile class of queueing networks with a computationally efficient product form solution, which had been proved to exist by using new techniques from the theory of fixed
point equation \cite{GN3}. In its simplest version, a signal arriving at a non empty queue forces a positive customer to leave the network immediately \cite{GN1}. Since their introduction, G-networks have been extensively studied covering several extensions such as triggered movement,
which redirect customers among the queues \cite{GN4}; catastrophes or batch service \cite{GN5}, adders \cite{GN6}; multiple classes of positive customers and signals \cite{gelenbe1998g}, state-dependent service disciplines \cite{four1,hende,chao}, tandem
networks \cite{gom1,har2}, deletion of a random amount of work \cite{box1,jain}, retrials \cite{art1,art2} (not exhaustive list). For a complete bibliography see \cite{gelre,do,cag}.

G-networks have been shown to be a diverse application tool to analyse and optimise the effects of dynamic load balancing in large scale networks \cite{Morfo} as well as in Gene Regulatory Networks \cite{G2007,KimG12}. A recent application of G-Networks is to the modelling of systems which operate with intermittent sources of energy, known as Energy Packet Networks \cite{EPN,Marin,fourn,Ceran2,dimit,zhang,NOLTA,kad,Kadi,omer,yin}. Other applications include their of approximating continuous and bounded real-valued functions \cite{Approx}, which serves as the foundation for learning algorithms \cite{gelrnn} and Deep Learning \cite{Yin2}. They have been also used for modelling natural neuronal networks \cite{Biosystems}, in image processing applications \cite{MRI,Cramer2}, and as a tool for predicting the toxicity of chemical compounds \cite{ICANN18}.

In the field of computer network performance, the RNN has been used to build distributed controllers for quality of service routing in packet networks \cite{CPN,Brun} and in the design of Software Defined Network controllers \cite{Francois1}. Real-time optimised task allocation algorithms in Cloud systems \cite{Wang2} have also been built and tested. Recent applications has addressed the use of the RNN to detect attacks on Internet of Things (IoT) gateways \cite{IoT}.
\paragraph{Stability and Delay analysis in Random Access Networks:}
The simple and decentralized nature of ALOHA protocol \cite{Abramson} made it very popular in multiple access communication systems. The ever increasing need for massive uncoordinated access has increased the interest on random access protocols \cite{Laya2014, METISComMag2014}, which remain an active research area with challenging open problems even for very simple networks \cite{AEUnion,TongSingProcMag2004}. 

In spite of its simple operation, stable throughput and delay analysis are quite challenging due to the high level of interaction among queues. This is due to the fact that each transmitted source interferes with other users with non-empty queues, i.e., the successful transmission probability is a function of the status of the other users. As a consequence, the departure rate of a queue can be computed only when we know the stationary distribution of the joint queue length process \cite{Rao_TIT1988}. This is the reason why the vast majority of previous works has focused on small-sized networks and only bounds or approximations are known for the networks with larger number of sources \cite{Tsybakov79,Rao_TIT1988,szpa1,LuoAE1999,NawareTong2005}. Delay analysis of random access networks is even more challenging both by mathematical and the application point of view. The high level of interaction among queues is the reason for the limited number of analytical results. For the sake of clarity we mention the works in \cite{NawareTong2005, szpa2, nain, BehrooziRaoTIT1992, GeorgiadisJSAC87} (not exhaustive list).

Recently, the authors have performed considerable contribution both in the investigation of stable throughput region and in the delay analysis by considering sophisticated \textit{queue-aware} transmission protocols in modern random access schemes \cite{dimpaptwc,dimpap}, in IoT networks \cite{ChenWoWMoM2016, ChenICCW2017, itc2018, ChenTWC2018} and in network-level cooperative wireless networks \cite{PappasTWC2015, adhoc}. 

\subsection{Contributions}
In this work, we consider a two-user slotted ALOHA multiple access channel with collisions, i.e., a transmission is successful if and only if a single user transmits. Each user has external bursty arrivals that are stored in its infinite-sized capacity queue and accesses the medium in a random access manner. The major contribution of this work relies on the fact that we consider for the first time in the related literature the concept of \textit{signals} in random access networks, which introduces the so-called \textit{Random Access G-Network}. 

In particular, we consider both \textit{negative} signals, which delete a packet from a user's queue, and signals that \textit{trigger} the instantaneous movement of packets among user queues. The concept of negative signals can model virus attacks in such networks or other possible malfunctions, whereas, triggering signals can be used to model load balancing schemes. To the best of our knowledge this variation of random access has not been reported so far. For such a network, we investigate stability conditions, the stable throughput region and the queueing delay. 
\subsubsection{Stability and Stable Throughput Analysis}
The presence of negative signals affects the stable throughput region, which now becomes a subset of the stability region, since packets can be dropped before being transmitted to the destination thus, these packets they do not contribute to the achievable throughput. Note that such property has never been reported in the literature of random access networks so far. 

When characterizing the stability and/or stable throughput in such a network, we have to cope with the problem of highly interaction among queues. In RAG network, the level of interaction is even higher due to the presence of signals that interact upon arrivals both with the queue and with the backlogged packets. We show that the stochastic dominant technique \cite{Rao_TIT1988} is still an efficient tool to bypass the problem of interaction among queues. Although stability conditions can be also derived by well known methods from two dimensional Markov chains \cite{fay,rw}, the importance of stochastic dominant technique is more apparent when we characterize the stable throughput region, where we show that it is a subset of the stability region.

\subsubsection{Queueing Analysis}
Queueing theoretic analysis of interacting queues is a quite challenging task. The increased level of interaction due to the presence of signals will further complicate the analysis. In this work, we present a compact mathematical analysis, and we provide exact expressions for the probability generating function (pgf) of the joint stationary queue length distribution of user queues by solving a non-homogeneous \textit{Riemann boundary value problem} \cite{ga}. Furthermore, we also provide a computationally efficient alternative way to obtain explicit bounds for basic performance metrics without calling for advanced mathematical concepts of boundary value problems.   

The rest of the paper is summarized as follows. In Section \ref{sec:model} we describe in detail the mathematical model, while in Section \ref{sec:SCSTR} we provide the stability and the stable throughput region. The fundamental functional equation along with some basic preliminary results regarding its investigation are presented in Section \ref{sec:funct}. In Section \ref{sec:bvp} we provide expressions for the pgf of the joint stationary queue length distribution at user queues in terms of a solution of a Riemann boundary value problem, while in Section \ref{sec:bounds} we provide exact bounds for the expected number of backlogged packets at each user queue, without using the theory of boundary value problems. Numerical results are given in Section \ref{sec:num} and show insights in system performance. Some conclusions and future directions are presented in Section \ref{sec:con}.

\section{The Mathematical Model}\label{sec:model}
We consider an ALOHA-type random access network (RAN) consisting of two users, say $U_{1}$, $U_{2}$, that communicate with a common destination node. Each user is equipped with infinite capacity buffers, in which it stores the arriving and backlogged packets. The time is slotted and packets have equal length. The transmission time of a packet corresponds to a single slot. 

At the beginning of a slot, user $U_{k}$, $k=1,2,$ if it is not-empty, transmits a packet with probability $\alpha_{k}$, or remains silent with probability $\bar{\alpha}_{k}=1-\alpha_{k}$. We consider a collision based channel model, and thus, if both users attempt to transmit simultaneously, a collision occurs and both transmissions fail. In such case, the packets have to be re-transmitted in a latter slot. 

Contrary to the traditional RAN, in this work we introduce a generalized RAN, called the Random Access G-Network (RAGN) by including the concept of signals. In particular, at the beginning of a slot, signals are generated in user $U_{k}$ with probability $s_{k}$. If a signal is generated, then with probability $l_{k}^{-}$ it deletes a packet from the buffer of $U_{k}$, while with probability $l_{k}^{+}=1-l_{k}^{-}$ it triggers the instantaneously movement of a packet from the buffer of $U_{k}$ to the buffer of the other user. A signal arriving in an empty buffer has no effect. In this work, we assume that if a signal and a packet transmission occur simultaneously, the signal occurs first. \textit{In case of a signal generation, a temporarily network malfunction occurs, and \textbf{both} users remain silent during the slot}. Moreover, packet arrivals are scheduled at the end of the slot, early departure late arrival model. Denote by $A_{k,n}$ the number of packets arriving in $U_{k}$, in the time interval $(n, n+1]$, with $\mathbb{E}(A_{n,k})=\lambda_{k}<\infty$, $k=1,2$. Denote the joint pgf of $A_{1,n}$ and $A_{2,n}$ by $H(x,y)=\mathbb{E}(x^{A_{1,n}}y^{A_{2,n}})$, $|x|\leq1$, $|y|\leq1$, $n\in\mathbb{N}$.

Let $Q_{k,n}$ be the queue length at the buffer of $U_{k}$ at the
beginning of time slot $n$. Under the above system model, the two-dimensional process $\mathbf{Q}_{n}=\{(Q_{1,n},Q_{2,n});n=0,1,\ldots\}$ is a Markov
chain. Clearly, under usual assumptions $\mathbf{Q}$ is irreducible and aperiodic. Then,
\begin{itemize}
\item If $Q_{1,n}=Q_{2,n}=0$, then $Q_{k,n+1}=A_{k,n}$, $k=1,2$.
\item If $Q_{1,n}=0$, $Q_{2,n}>0$,
\begin{displaymath}
\begin{array}{rl}
Q_{1,n+1}=&\left\{\begin{array}{ll}
A_{1,n},&w.p.\ \bar{s}_{2}+s_{2}l_{2}^{-},\\
A_{1,n}+1,&w.p.\ s_{2}l_{2}^{+},
\end{array}\right.\vspace{2mm}\\
Q_{2,n+1}=&\left\{\begin{array}{ll}
Q_{2,n}+A_{2,n},&w.p.\ \bar{s}_{2}\bar{\alpha}_{2},\\
Q_{2,n}-1+A_{2,n},&w.p.\ \bar{s}_{2}\alpha_{2}+s_{2},
\end{array}\right.
\end{array}
\end{displaymath}
\item If $Q_{1,n}>0$, $Q_{2,n}=0$,
\begin{displaymath}
\begin{array}{rl}
Q_{1,n+1}=&\left\{\begin{array}{ll}
Q_{1,n}+A_{1,n},&w.p.\ \bar{s}_{1}\bar{\alpha}_{1},\\
Q_{1,n}-1+A_{1,n},&w.p.\ \bar{s}_{1}\alpha_{1}+s_{1},
\end{array}\right.\vspace{2mm}\\
Q_{2,n+1}=&\left\{\begin{array}{ll}
A_{2,n},&w.p.\ \bar{s}_{1}+s_{1}l_{1}^{-},\\
A_{2,n}+1,&w.p.\ s_{1}l_{1}^{+},
\end{array}\right.
\end{array}
\end{displaymath}
\item If $Q_{1,n}>0$, $Q_{2,n}>0$,
\begin{displaymath}
\begin{array}{rl}
(Q_{1,n+1}Q_{2,n+1})=&\left\{\begin{array}{ll}
(Q_{1,n}+A_{1,n},Q_{2,n}+A_{2,n}),&w.p.\ \bar{s}_{1}\bar{s}_{2}(\bar{\alpha}_{1}\bar{\alpha}_{2}+\alpha_{1}\alpha_{2})+s_{1}s_{2}l_{1}^{+}l_{2}^{+},\\
(Q_{1,n}-1+A_{1,n},Q_{2,n}+A_{2,n}),&w.p.\ \bar{s}_{1}\bar{s}_{2}\alpha_{1}\bar{\alpha}_{2}+s_{1}\bar{s}_{2}l_{1}^{-}+s_{1}s_{2}l_{1}^{+}l_{2}^{-},\\
(Q_{1,n}-1+A_{1,n},Q_{2,n}+1+A_{2,n}),&w.p.\ s_{1}\bar{s}_{2}l_{1}^{+},\\
(Q_{1,n}+A_{1,n},Q_{2,n}-1+A_{2,n}),&w.p.\ \bar{s}_{1}\bar{s}_{2}\alpha_{2}\bar{\alpha}_{1}+s_{2}\bar{s}_{1}l_{2}^{-}+s_{1}s_{2}l_{2}^{+}l_{1}^{-},\\
(Q_{1,n}+1+A_{1,n},Q_{2,n}-1+A_{2,n}),&w.p.\ s_{2}\bar{s}_{1}l_{2}^{+},\\
(Q_{1,n}-1+A_{1,n},Q_{2,n}-1+A_{2,n}),&w.p.\ s_{1}s_{2}l_{1}^{1}l_{2}^{-}.
\end{array}\right.
\end{array}
\end{displaymath}
\end{itemize}

\section{Stability Conditions and Stable Throughput Region} \label{sec:SCSTR}

From the previous section, it is clear that for a packet located at the head of the queue, there are three options depending if a signal will be generated or not. When a signal is not generated, the packet it will be transmitted successfully to the destination if the node will attempt to transmit and collision will not happen. If a signal is generated, there are two options, either the packet will be dropped from the system or it will be transferred to the other queue. Thus, it is important to emphasize that the stability region is different than the stable throughput region. In fact, the stable throughput region is a subset of the stability region as we will see in this section.

The service probability for $U_i$ is denoted by $\mu_i$, $i=1,2$, which denotes the probability that a packet will be either transmitted successfully to the destination or it will be relocated to the other queue. The expression for $\mu_{1}$ is given by
\begin{equation} \label{eq:mu1}
\mu_{1}=\alpha_{1} \mathrm{Pr}(Q_2=0) \bar{s}_1 +\alpha_{1} \mathrm{Pr}(Q_2>0) \bar{s}_1 \bar{\alpha}_{2} \bar{s}_2+s_1l_{1}^{+}.
\end{equation}

Then we define the probability that a packet will be removed from a queue, either it will dropped, or transmitted successfully, or transferred to the other queue. For the users $U_k$, $k=1,2$, we denote this probability $m_k$. The expression for $m_1$ is given by

\begin{equation}
m_1= \mu_{1}+s_1l_{1}^{-}.
\end{equation}

which is equal to

\begin{equation} \label{eq:m1}
m_1= \alpha_{1} \bar{s}_1 \left[ \mathrm{Pr}(Q_2=0)+ \mathrm{Pr}(Q_2>0) \bar{\alpha}_{2}\bar{s}_2\right] +s_1.
\end{equation}

Similarly, we can write the expression for $m_2$.
Clearly, $m_1$ and $\mu_{1}$ depend on the state of the second queue and $m_2$ and $\mu_{2}$ on the state of the first queue. Thus, the queues are coupled. We will bypass this difficulty by applying the stochastic dominance technique to obtain the exact stability region and the stable throughput region.

We use the following definition of queue stability \cite{szpa,szpa1}:

\begin{definition}
Denote by $Q_i^t$ the length of queue $i$ at the beginning of time slot $t$. The queue is said to be \emph{stable} if
$\lim_{t \rightarrow \infty} {Pr}[Q_i^t < {x}] = F(x)$ and $\lim_{ {x} \rightarrow \infty} F(x) = 1$.
\end{definition}

Although we will not make explicit use of this definition we use its corollary consequence which is Loynes' theorem~\cite{Loynes} that states that if the arrival and service processes of a queue are strictly jointly stationary and the average arrival rate is less than the average service rate, then the queue is stable. If the average arrival rate is greater than the average service rate, then the queue is unstable and the value of $Q_i^t$ approaches infinity almost surely. The stability region of the system is defined as the set of arrival rate vectors $\boldsymbol{\lambda}=(\lambda_1, \lambda_2)$ for which the queues in the system are stable. In the previous definition, $\lambda_i$ denotes the total arrival rate at the queue $i$.

Before proceeding with the derivation of the stability region and the stable throughput region, we need to calculate the internal arrival probability, $\lambda_{j,i}$ for each queue $j=1,2$, $i$ stands for internal. We have that 

\begin{equation} \label{eq:l1igen}
\lambda_{1,i}=\mathrm{Pr}(Q_2>0) s_2 l_{2}^{+}.
\end{equation}

and

\begin{equation} \label{eq:l2igen}
\lambda_{2,i}=\mathrm{Pr}(Q_1>0) s_1 l_{1}^{+}.
\end{equation}

The following theorem provides the stability region for the considered system.

\begin{theorem}[Stability Region]\label{stabil}
The stability region $\mathcal{R}$ for a fixed transmission probability vector $\mathbf{\alpha}:=[\alpha_{1},\alpha_{2}]$ is given by $\mathcal{R}=\mathcal{R}_1 \cup \mathcal{R}_2$ where 
\begin{equation}
\begin{aligned} \label{stab_R1}
\mathcal{R}_1 = \left\{ (\lambda_{1},\lambda_{2}) :\lambda_{1}+\frac{(\lambda_{2}+s_1l_{1}^{+})(s_2l_{2}^{+}+\alpha_{1}\bar{s}_1(1-\bar{\alpha}_{2} \bar{s}_2))}{\alpha_{2} \bar{s}_2 \bar{\alpha}_1 \bar{s}_1 +s_2}<\alpha_{1} \bar{s}_1 + s_1,  \right. \\
\left.  \lambda_{2}+s_1 l_{1}^{+} <\alpha_{2} \bar{s}_2 \bar{\alpha}_{1}\bar{s}_1+s_2 \right\},
\end{aligned}
\end{equation}
\begin{equation}
\begin{aligned}\label{stab_R2}
\mathcal{R}_2 = \left\{ (\lambda_{1},\lambda_{2}) :\lambda_{2}+\frac{(\lambda_{1}+s_2l_{2}^{+})(s_1l_{1}^{+}+\alpha_{2}\bar{s}_2(1-\bar{\alpha}_{1} \bar{s}_1))}{\alpha_{1} \bar{s}_1 \bar{\alpha}_2 \bar{s}_2 +s_1}<\alpha_{2} \bar{s}_2 + s_2, \right. \\
\left.  \lambda_{1}+s_2 l_{2}^{+} <\alpha_{1} \bar{s}_1 \bar{\alpha}_{2}\bar{s}_2+s_1 \right\},
\end{aligned}
\end{equation}
\end{theorem}

\begin{proof}

To determine the stability region of our system we apply the stochastic dominance technique \cite{Rao_TIT1988}, i.e. we construct hypothetical dominant systems, in which a node transmits dummy packets for the packet queue that is empty, while for the non-empty queue transmits according to its traffic. Under this approach, we consider the $R_{1}$, and $R_{2}$-dominant systems. In the $R_{k}$ dominant system, whenever the queue $k$, $k=1,2$ empties, it continues transmitting a dummy packet. Thus, in $R_{1}$, node $1$ never empties, and hence, node $2$ sees a constant probability that a packet will be removed from its queue, while that probability for node $1$ depends on the state of node $2$, i.e., empty or not. We proceed with dominant system $R_1$. The probability $m_1$ of the first node is given by \eqref{eq:m1}. The $m_2$ is given by

\begin{equation} \label{eq:m2}
m_2= \alpha_{2} \bar{s}_2 \left[ \mathrm{Pr}(Q_1=0)+ \mathrm{Pr}(Q_1>0) \bar{\alpha}_{1}\bar{s}_1\right] +s_2.
\end{equation}

Since in $R_1$ queue $1$ never empties, $\mathrm{Pr}(Q_1>0)=1$, we have $m_2=\alpha_{2} \bar{s}_2 \bar{\alpha}_{1}\bar{s}_1+s_2$. We can obtain that the total arrival rate at $U_2$ is $\lambda_2+s_1 l_{1}^{+}$. By applying Loyne's criterion \cite{Loynes}, the second node is stable if and only if $\lambda_{2}+s_1 l_{1}^{+} < m_2$. The probability that the second node is empty and is given by 

\begin{equation} \label{eq:PrQ2empty}
\mathrm{Pr}(Q_2 = 0)=1-\frac{\lambda_{2}+s_1 l_{1}^{+}}{m_2}=1-\frac{\lambda_{2}+s_1 l_{1}^{+} }{\alpha_{2} \bar{s}_2 \bar{\alpha}_{1}\bar{s}_1+s_2}.
\end{equation}

Then, we can obtain the internal arrival probability at $U_1$ is given by

\begin{equation} \label{eq:l1R1}
\lambda_{1,i}=\frac{s_2 l_{2}^{+}(\lambda_{2}+s_1 l_{1}^{+})}{\alpha_{2} \bar{s}_2 \bar{\alpha}_{1}\bar{s}_1+s_2}.
\end{equation}

After replacing $\mathrm{Pr}(Q_2 = 0)$ into \eqref{eq:m1}, and applying Loynes criterion \cite{Loynes} we can obtain the stability condition for the first node. Then, we have the stability region $\mathcal{R}_1$ given by \eqref{stab_R1}. 
Similarly, we can obtain the stability region for the second dominant system $\mathcal{R}_2$, the proof is omitted. For a detailed treatment of dominant systems please refer to \cite{Rao_TIT1988}.

An important observation made in \cite{Rao_TIT1988} is that the stability conditions obtained by the stochastic dominance technique are not only sufficient but also necessary for the stability of the original system. The \emph{indistinguishability} argument \cite{Rao_TIT1988} applies here as well. Based on the construction of the dominant system, we can see that the queue sizes in the dominant system are always greater than those in the original system, provided they are both initialized to the same value and the arrivals are identical in both systems. Therefore, given $\lambda_{2}<m_{2}$, if for some $\lambda_{1}$, the queue at the first user is stable in the dominant system, then the corresponding queue in the original system must be stable. Conversely, if for some $\lambda_{1}$ in the dominant system, the queue at the first node saturates, then it will not transmit dummy packets, and as long as the first user has a packet to transmit, the behavior of the dominant system is identical to that of the original system since dummy packet transmissions are eliminated as we approach the stability boundary. Therefore, the original and the dominant system are indistinguishable at the boundary points.
\end{proof}

\begin{theorem}[Stable Throughput Region] \label{Thr:STR}
The stable throughput region $\mathcal{T}$ for a fixed transmission probability vector $\mathbf{\alpha}:=[\alpha_{1},\alpha_{2}]$ is given by $\mathcal{T}=\mathcal{T}_1 \cup \mathcal{T}_2$ where 
\begin{equation}
\begin{aligned} \label{stab_T1}
\mathcal{T}_1 = \left\{ (\lambda_{1},\lambda_{2}) :\lambda_{1}+\frac{(\lambda_{2}+s_1l_{1}^{+})(s_2l_{2}^{+}+\alpha_{1}\bar{s}_1(1-\bar{\alpha}_{2} \bar{s}_2))}{\alpha_{2} \bar{s}_2 \bar{\alpha}_1 \bar{s}_1 +s_2}<\alpha_{1} \bar{s}_1 + s_1l_{1}^{+}, \right. \\
\left. \lambda_{2}+s_1 l_{1}^{+} <\alpha_{2} \bar{s}_2 \bar{\alpha}_{1}\bar{s}_1+s_2l_{2}^{+} \right\},
\end{aligned}
\end{equation}
\begin{equation}
\begin{aligned}\label{stab_T2}
\mathcal{T}_2 = \left\{ (\lambda_{1},\lambda_{2}) :\lambda_{2}+\frac{(\lambda_{1}+s_2l_{2}^{+})(s_1l_{1}^{+}+\alpha_{2}\bar{s}_2(1-\bar{\alpha}_{1} \bar{s}_1))}{\alpha_{1} \bar{s}_1 \bar{\alpha}_2 \bar{s}_2 +s_1}<\alpha_{2} \bar{s}_2 + s_2l_{2}^{+}, \right. \\
\left. \lambda_{1}+s_2 l_{2}^{+} <\alpha_{1} \bar{s}_1 \bar{\alpha}_{2}\bar{s}_2+s_1l_{1}^{+} \right\},
\end{aligned}
\end{equation}
\end{theorem}

\begin{proof}
Following the same methodology of dominant systems in the previous theorem, we construct two hypothetical systems, $R_1$ and $R_2$. In $R_1$, $U_1$ transmits dummy packets when its queue is empty, thus $\mathrm{Pr}(Q_1>0)=1$, we have $\mu_2=\alpha_{2} \bar{s}_2 \bar{\alpha}_{1}\bar{s}_1+s_2l_{2}^{+}$. In the proof of Theorem \ref{stabil}, we obtain $\mathrm{Pr}(Q_2>0)$ and is given by \eqref{eq:PrQ2empty}, and the internal arrival probability for $U_1$ is given by \eqref{eq:l1R1}.
After replacing \eqref{eq:PrQ2empty} into \eqref{eq:mu1} and applying Loyne's theorem \cite{Loynes} we have that the stable throughput region obtained from $R_1$ is given by \eqref{stab_T1}. Similarly we obtain $\mathcal{T}_2$ which is given by \eqref{stab_T2}.
\end{proof}

\begin{remark}
Clearly, $\mathcal{T} \subset \mathcal{R}$ this can be easily observed by the expressions in Theorems \ref{stabil} and \ref{Thr:STR}. If $s_i>0$ and $l_{i}^{+}=1$, then $\mathcal{T} = \mathcal{R}$. If $s_i=0$ for all $i$ then $\mathcal{T} = \mathcal{R}$ and the stability region is the same with the two-user multiple access channel with collision reported in the literature in \cite{Rao_TIT1988}.
\end{remark}

\begin{remark}
In Appendix \ref{stab} we provide an alternative way to derive the stability condition based on general results regarding two-dimensional random walks \cite{rw}.
\end{remark}

\section{The Functional Equation and Preliminary Analysis}\label{sec:funct}
In the following, we apply the generating function approach and obtain a fundamental functional equation, which is the key element of our analysis. In order to proceed with the investigation of the functional equation we also need some crucial preliminary results that are also given below. To proceed, the evolution among queues implies
\begin{equation}
\begin{array}{rl}
\mathbb{E}(x^{Q_{1,n+1}}y^{Q_{2,n+1}})=&H(x,y)\left[\mathbb{E}(\mathbf{1}_{\{Q_{1,n}=Q_{2,n}=0}\})\right.\vspace{2mm}\\
&\left.+[s_{1}(\frac{l_{1}^{-}}{x}+\frac{l_{1}^{+}y}{x})+\bar{s}_{1}(\bar{\alpha}_{1}+\frac{\alpha_{1}}{x})]\mathbb{E}(x^{Q_{1,n}}\mathbf{1}_{\{Q_{1,n}>0,Q_{2,n}=0}\})\right.\vspace{2mm}\\
&\left.+[s_{2}(\frac{l_{2}^{-}}{y}+\frac{l_{2}^{+}x}{y})+\bar{s}_{2}(\bar{\alpha}_{2}+\frac{\alpha_{2}}{y})]\mathbb{E}(y^{Q_{2,n}}\mathbf{1}_{\{Q_{1,n}=0,Q_{2,n}>0}\})\right.\vspace{2mm}\\
&\left.+\left(s_{1}\bar{s}_{2}(\frac{l_{1}^{-}}{x}+\frac{l_{1}^{+}y}{x})+\bar{s}_{1}s_{2}(\frac{l_{2}^{-}}{y}+\frac{l_{2}^{+}x}{y})+s_{1}s_{2}(\frac{l_{1}^{-}l_{2}^{-}}{xy}+\frac{l_{1}^{-}l_{2}^{+}}{y}+\frac{l_{1}^{+}l_{2}^{-}}{x}+l_{1}^{+}l_{2}^{+})\right.\right.\vspace{2mm}\\
&\left.\left.+\bar{s}_{1}\bar{s}_{2}(\bar{\alpha}_{1}\bar{\alpha}_{2}+\alpha_{1}\alpha_{2}+\frac{\alpha_{1}\bar{\alpha}_{2}}{x}+\frac{\alpha_{2}\bar{\alpha}_{1}}{y})\right)\mathbb{E}(x^{Q_{1,n}}y^{Q_{2,n}}\mathbf{1}_{\{Q_{1,n}>0,Q_{2,n}>0}\})\right].
\end{array}
\end{equation}

Assuming that the system is stable (see Section \ref{sec:SCSTR}), letting
\begin{displaymath}
\Pi(x,y)=\lim_{n\to\infty}\mathbb{E}(x^{Q_{1,n}}y^{Q_{2,n}}),\,|x|\leq1,|y|\leq1,
\end{displaymath}
and having in mind that
\begin{displaymath}
\begin{array}{rl}
\lim_{n\to\infty}\mathbb{E}(\mathbf{1}_{\{Q_{1,n}=Q_{2,n}=0}\})=&\Pi(0,0),\\
\lim_{n\to\infty}\mathbb{E}(x^{Q_{1,n}}\mathbf{1}_{\{Q_{1,n}>0,Q_{2,n}=0}\})=&\Pi(x,0)-\Pi(0,0),\\
\lim_{n\to\infty}\mathbb{E}(y^{Q_{2,n}}\mathbf{1}_{\{Q_{1,n}=0,Q_{2,n}>0}\})=&\Pi(0,y)-\Pi(0,0),\\
\lim_{n\to\infty}\mathbb{E}(x^{Q_{1,n}}y^{Q_{2,n}}\mathbf{1}_{\{Q_{1,n}>0,Q_{2,n}>0}\})=&\Pi(x,y)-\Pi(x,0)-\Pi(0,y)+\Pi(0,0),
\end{array}
\end{displaymath}

we obtain the following functional equation:
\begin{equation}
\begin{array}{rl}
\Pi(x,y)[xy-\phi_{3}(x,y)]=[y\phi_{1}(x,y)-\phi_{3}(x,y)][\Pi(x,0)-\Pi(0,0)]\vspace{2mm}\\+[x\phi_{2}(x,y)-\phi_{3}(x,y)][\Pi(0,y)-\Pi(0,0)]+[xyH(x,y)-\phi_{3}(x,y)]\Pi(0,0),\end{array}
\label{fun}
\end{equation}
where
\begin{displaymath}
\begin{array}{rl}
\phi_{3}(x,y)=&H(x,y)[\bar{s}_{1}\bar{s}_{2}(xy(\bar{\alpha}_{1}\bar{\alpha}_{2}+\alpha_{1}\alpha_{2})+\alpha_{1}\bar{\alpha}_{2}y+\bar{\alpha}_{1}\alpha_{2}x)+s_{1}\bar{s}_{2}y(l_{1}^{-}+l_{1}^{+}y)\\
&+s_{2}\bar{s}_{1}x(l_{2}^{-}+l_{2}^{+}x)+s_{1}s_{2}(l_{1}^{-}l_{2}^{-}+l_{1}^{+}l_{2}^{-}y+l_{1}^{-}l_{2}^{+}x+l_{1}^{+}l_{2}^{+}xy)],\\
\phi_{1}(x,y)=&H(x,y)[\bar{s}_{1}(\bar{\alpha}_{1}x+\alpha_{1})+s_{1}(l_{1}^{-}+l_{1}^{+}y)],\\
\phi_{2}(x,y)=&H(x,y)[\bar{s}_{2}(\bar{\alpha}_{2}y+\alpha_{2})+s_{2}(l_{2}^{-}+l_{2}^{+}x)].
\end{array}
\end{displaymath}
The function $Z(x,y):=xy-\phi_{3}(x,y)$, $|x|\leq 1$, $|y|\leq 1$ is called the \textit{kernel} of the functional equation (\ref{fun}), and plays a crucial role in its solution procedure.

The definition of $\Pi(x,y)$, implies that for fixed $y$, with $|y|\leq 1$, $\Pi(x,y)$ is regular in $x$, with $|x|<1$, continuous in $x$, with $|x|\leq 1$, and similarly with $x$, $y$ interchanged, and the coefficients of $x^{i}y^{j}$, $(i,j)\in\mathbb{N}_{0}\times\mathbb{N}_{0}=\{0,1,\ldots\}\times\{0,1,\ldots\},$ in the series expansion of $\Pi(x,y)$ are non-negative, with $\Pi(1,1)=1$. Then, it also follows that $\Pi(p,0)$, $\Pi(0,p)$ are both regular in $p$, with $|p|<1$, continuous in $p$, with $|p|\leq 1$, and the coefficients in the series expansion of $\Pi(p,0)$ are all non-negative. Similar result holds for $\Pi(0,p)$.

Note that some interesting relations can be directly derived by the functional equation (\ref{fun}). More precisely, setting $y=1$, rearrange its terms and then taking the limit $x\to 1$, and vice versa we come up with the following ``conservation of flow relations":
\begin{equation}
\begin{array}{rl}
\lambda_{1}+s_{2}l_{2}^{+}(1-\Pi(1,0))=&(s_{1}+\bar{s}_{1}\bar{s}_{2}\alpha_{1}\bar{\alpha}_{2})[1-\Pi(1,0)-\Pi(0,1)+\Pi(0,0)]\\
&+(s_{1}+\bar{s}_{1}\alpha_{1})[\Pi(1,0)-\Pi(0,0)],\\
\lambda_{2}+s_{1}l_{1}^{+}(1-\Pi(0,1))=&(s_{2}+\bar{s}_{1}\bar{s}_{2}\alpha_{2}\bar{\alpha}_{1})[1-\Pi(1,0)-\Pi(0,1)+\Pi(0,0)]\\
&+(s_{2}+\bar{s}_{2}\alpha_{2})[\Pi(0,1)-\Pi(0,0)].
\end{array}
\label{con}
\end{equation}

Note that (\ref{con}) has a clear probabilistic interpretation. In particular, they equate the flow of packets into the buffer of $U_{k}$, with the flow of jobs out of the buffer of $U_{k}$, $k=1,2$, respectively. Just for the sake of clarity, note that the left hand side in the first of (\ref{con}) is composed of the mean number of external arrivals per slot in the buffer of $U_{1}$ (i.e. $\lambda_{1}$), while the second term corresponds to the arrivals per slot due to the presence of signals in $U_{2}$ that trigger the instantaneous transfer of a packet in $U_{1}$. Note that such an arrival may occur only in case $U_{2}$ is not empty, i.e., with probability $1-\Pi(1,0)$, which is the fraction of time user $U_{2}$ is not empty. 
 
\subsection{On the kernel $Z(x,y)$}
To proceed with the analysis of functional equations of type (\ref{fun}), it is required to focus on the investigation of \textit{suitable} representations of the zero-tuples of its kernel $Z(x,y)$. For a complete discussion the interested reader is referred to \cite{coh1,coh2,fay1}\footnote{Note that in case of Bernoulli arrivals, $Z(x,y)$ is a quadratic polynomial with respect to $x$, $y$.}

Clearly, the definition of $Z(x,y)$, implies that for fixed $y$, with $|y|\leq 1$, $Z(x,y)$ is regular in $x$, with $|x|<1$, continuous in $x$, with $|x|\leq 1$, and similarly with $x$, $y$ interchanged. In the following we assume \textit{symmetrical} characterization of the zero tuples $(\widehat{x},\widehat{y})$ of the kernel \cite{rw,coh1}, by letting
\begin{displaymath}
\widehat{x}=gr,\,\widehat{y}=gr^{-1},\,|r|=1, |g|\leq1,
\end{displaymath}
and thus,
\begin{displaymath}
Z(\widehat{x},\widehat{y})=g^{2}-\phi_{3}(gr,gr^{-1})=0.
\end{displaymath}
Then, for $|r|=1$, $r\neq1$, $|g|=1$,
\begin{displaymath}
\begin{array}{l}
|\phi_{3}(gr,gr^{-1})|\leq|\bar{s}^{2}(g^{2}(\bar{\alpha}^{2}+\alpha^{2})+\alpha\bar{\alpha}g(r+r^{-1}))+s\bar{s}gr^{-1}(l^{-}+l^{+}gr^{-1})\\
+s\bar{s}gr(l^{-}+l^{+}gr)+s^{2}((l^{-})^{2}+l^{+}l^{-}g(r+r^{-1})+(l^{+})^{2}g^{2})|<1=|g|^{2}.
\end{array}
\end{displaymath}
Since $\phi_{3}(gr,gr^{-1})$ is regular in $g$ for $|g| < 1$, with fixed $r$, $|r| = 1$, and continuous for $|g| \leq 1$, Rouch\'e's theorem implies that $g^{2}-\phi_{3}(gr,gr^{-1})=0$ has exactly two roots in $|g|\leq1$. According to Theorem 2.1, pp. 65-66, \cite{rw}, $g(r)$ is such that $g(r)\to 1$, as $r\to 1$, $|r|=1$. Put,
\begin{displaymath}
\begin{array}{c}
\mathcal{S}_{1}=\{x:x=g(r)r,|r|=1\},\,\mathcal{S}_{2}=\{y:y=g(r)r^{-1},|r|=1\}.
\end{array}
\end{displaymath}

In order to proceed, we must show that the contours $\mathcal{S}_{1}$, $\mathcal{S}_{2}$ are both simple and smooth. To show this in the general case, we need some extra assumptions (see \cite{coh1}, Sections II.3.1-II.3.3, pp. 151-163) that will further complicate the analysis and worse the readability of the paper. 

With that in mind, we will focus on the \textit{symmetrical} case (see next Section). For such a case Theorem 2.1, pp. 65-66 in \cite{rw} implies that:
\begin{enumerate}
\item $\mathcal{S}_{1}$, $\mathcal{S}_{2}$ are both simple and smooth.
\item $x=0\in \mathcal{S}_{1}^{+}$, $x=\infty\in\mathcal{S}_{1}^{-}$, $y=0\in\mathcal{S}_{2}^{+}$, $y=\infty\in\mathcal{S}_{2}^{-}$, where $\mathcal{S}_{k}^{+}$ (resp. $\mathcal{S}_{k}^{-}$) denotes the interior (resp. the exterior) of $\mathcal{S}_{k}$, $k=1,2$, and $0<g(r)\leq1$, $|r|=1$.
\item $x\in\mathcal{S}_{1}$ (resp. $y\in\mathcal{S}_{2}$), then $\bar{x}\in\mathcal{S}_{1}$ (resp. $\bar{y}\in\mathcal{S}_{2}$), where $\bar{h}$ denotes the complex conjugate of $h$.
\item the set $(\widehat{x},\widehat{y})\in\mathcal{S}_{1}\times\mathcal{S}_{2}$ of the zero tuples of the kernel $Z(x,y)=0$, and when $\widehat{x}$ traverses $\mathcal{S}_{1}$ counterclockwise, $\widehat{y}$, traverses $\mathcal{S}_{2}$ clockwise.
\end{enumerate}
Therefore, the solution of (\ref{fun}) is now reduced to the construction of functions $\Pi(p,0)$, $\Pi(0,p)$, both regular in $|p|<1$, which satisfy
\begin{equation}
\begin{array}{l}
\widehat{y}[\widehat{x}-\phi_{1}(\widehat{x},\widehat{y})][\Pi(\widehat{x},0)-\Pi(0,0)]+\widehat{x}[\widehat{y}-\phi_{2}(\widehat{x},\widehat{y})][\Pi(0,\widehat{y})-\Pi(0,0)]=\Pi(0,0)[H(\widehat{x},\widehat{y})-1],\end{array}
\label{fa}
\end{equation}
for all $(\widehat{x},\widehat{y})\in\mathcal{S}_{1}\times\mathcal{S}_{2}$. This problem will be transformed in the next section into a \textit{Riemann boundary value problem}. 
\section{The Symmetrical System: Exact Analysis}\label{sec:bvp}
As for the symmetrical system we assume $s_{k}=s$, $\lambda_{k}=\lambda$, $\alpha_{k}=\alpha$, $l_{k}^{+}=l^{+}$, $l_{k}^{-}=l$, $k=1,2$. Note that in such a case, $\phi_{3}(x,y)=\phi_{3}(y,x)$, and $\phi_{1}(x,y)=\phi_{2}(y,x)$. Moreover, for the stability conditions (see Theorem \ref{stabil} or Theorem \ref{stab1}) the following conditions should hold:
\begin{equation}
\begin{array}{rl}
\lambda+sl^{+}<&s+\bar{s}^{2}\alpha\bar{\alpha},\\
\lambda+sl^{+}\frac{\lambda+sl^{+}}{s+\bar{s}^{2}\alpha\bar{\alpha}}<&(s+\bar{s}\alpha)(1-\frac{\lambda+sl^{+}}{s+\bar{s}^{2}\alpha\bar{\alpha}})+\frac{\lambda+sl^{+}}{s+\bar{s}^{2}\alpha\bar{\alpha}}(s+\bar{s}^{2}\alpha\bar{\alpha}).
\end{array}
\end{equation}
Set,
\begin{displaymath}
\mathcal{C}:=\{z:|z|=1\},\,\mathcal{C}^{+}:=\{z:|z|<1\},\,\mathcal{C}^{-}:=\{z:|z|>1\}.
\end{displaymath}
We proceed as in \cite{coh1}, sections II.2.4, II.2.13, II.2.16; see also Theorem 3.2.1 in \cite{rw} or Theorem 1.1 in \cite{coh2}. In particular, for the contours $\mathcal{S}_{1}$, $\mathcal{S}_{2}$, and due to the symmetry of our system the following hold:
\begin{enumerate}
\item $z=0\in\mathcal{C}^{+}$, $z=1\in\mathcal{C}$, $z=\infty\in\mathcal{C}^{-}$,
\item and functions
\begin{displaymath}
x=x(z):\mathcal{C}^{+}\cup\mathcal{C}\to\mathcal{S}_{1}^{+}\cup\mathcal{S}_{1},\,\,y=y(z):\mathcal{C}^{-}\cup\mathcal{C}\to\mathcal{S}_{2}^{+}\cup\mathcal{S}_{2},
\end{displaymath} 
where,
\begin{enumerate}
\item $x(0)=0$, $x(1)=1=y(1)$, $y(\infty)=0$,
\item $x(z):\mathcal{C}^{+}\to\mathcal{S}_{1}^{+}$ is regular and univalent for $z\in \mathcal{C}^{+}$.
\item $y(z):\mathcal{C}^{-}\to\mathcal{S}_{2}^{+}$ is regular and univalent for $z\in \mathcal{C}^{-}$.
\end{enumerate}
\end{enumerate} 

Then, theorem 3.2.3, p. 123 in \cite{rw} states that $x(z)$, $z\in\mathcal{C}\cup\mathcal{C}^{+}$, and $y(z)$, $z\in\mathcal{C}\cup\mathcal{C}^{-}$, are such that
\begin{equation}
\begin{array}{rl}
x(z)=&\left\{\begin{array}{ll}
ze^{\frac{1}{2\pi i}\int_{|\zeta|=1}log\{g(e^{i\lambda(\zeta)})\}\{\frac{\zeta+z}{\zeta-z}-\frac{\zeta+1}{\zeta-1}\}\frac{d\zeta}{\zeta}},&z\in\mathcal{C}^{+},\\
g(e^{i\lambda(z)})e^{i\lambda(z)},&z\in\mathcal{C},
\end{array}\right.\vspace{2mm}\\
y(z)=&\left\{\begin{array}{ll}
z^{-1}e^{-\frac{1}{2\pi i}\int_{|\zeta|=1}log\{g(e^{i\lambda(\zeta)})\}\{\frac{\zeta+z}{\zeta-z}-\frac{\zeta+1}{\zeta-1}\}\frac{d\zeta}{\zeta}},&z\in\mathcal{C}^{-},\\
g(e^{i\lambda(z)})e^{-i\lambda(z)},&z\in\mathcal{C},
\end{array}\right.
\end{array}
\label{conf}
\end{equation}
where $\lambda(z)\in[0,2\pi)$, $z\in\mathcal{C}$, $\lambda(1)=0$, be the unique solution of
\begin{equation}
e^{i\lambda(z)}=ze^{\frac{1}{2\pi i}\int_{|\zeta|=1}log\{g(e^{i\lambda(\zeta)})\}\{\frac{\zeta+z}{\zeta-z}-\frac{\zeta+1}{\zeta-1}\}\frac{d\zeta}{\zeta}},\,z\in\mathcal{C}.
\label{conf1}
\end{equation}

For such a case $x(z)$ (resp. $y(z)$) represents a conformal mapping from $\mathcal{C}^{+}$ (resp. $\mathcal{C}^{-}$) onto $\mathcal{S}_{1}^{+}$ (resp. $\mathcal{S}_{2}^{+}$), while for every $z$, with $z\in\mathcal{C}$, $(\widehat{x},\widehat{y})=(x(z),y(z))$ is a zero tuple of the kernel $Z(x,y)$.

Since for $z\in\mathcal{C}$, $(\widehat{x},\widehat{y})=(x(z),y(z))$ is such that $x(z)y(z)=\phi_{3}(x(z),y(z))$, (\ref{fa}) implies after some algebra that
\begin{equation}
\Phi_{1}(z)=\Phi_{2}(z)G(z)+S(z),
\label{fq}
\end{equation}
where,
\begin{displaymath}
\begin{array}{rl}
\Phi_{1}(z)=&\frac{\Pi(x(z),0)-\Pi(0,0)}{\Pi(0,0)},\vspace{2mm}\\
\Phi_{2}(z)=&\frac{\Pi(0,y(z))-\Pi(0,0)}{\Pi(0,0)},\vspace{2mm}\\
G(z)=&\frac{x(z)[y(z)-\phi_{2}(x(z),y(z))]}{y(z)[\phi_{1}(x(z),y(z))-x(z)]},\vspace{2mm}\\
S(z)=&\frac{x(z)[1-H(x(z),y(z))]}{\phi_{1}(x(z),y(z))-x(z)}.
\end{array}
\end{displaymath}
Note that due to the regularity of $x(z)$, $z\in\mathcal{C}^{+}$, and of $y(z)$, $z\in\mathcal{C}^{-}$,
\begin{enumerate}
\item $\Phi_{1}(z)$, $z\in\mathcal{C}^{+}\cup\mathcal{C}$, is regular for $z\in\mathcal{C}^{+}$, and continuous for $z\in\mathcal{C}^{+}\cup\mathcal{C}$.
\item $\Phi_{2}(z)$, $z\in\mathcal{C}^{-}\cup\mathcal{C}$, is regular for $z\in\mathcal{C}^{-}$, and continuous for $z\in\mathcal{C}^{-}\cup\mathcal{C}$.
\end{enumerate}
Note also that $\Phi_{1}(z)$ is well defined due to the fact that $|x(z)|\leq 1$, $z\in\mathcal{C}$, so that $|x(z)|< 1$, $z\in\mathcal{C}^{+}$ due to the regularity of $x(z)$, $z\in\mathcal{C}^{+}$ and the maximum modulus theorem \cite{neh} (similar result holds for $\Phi_{2}(z)$), and $\Phi_{1}(0)=\lim_{|z|\to\infty}\Phi_{2}(z)=0$. 

From the above discussion, $\Phi_{1}(z)$, $\Phi_{2}(z)$, can be obtained as a solution of a non-homogeneous Riemann boundary value problem, with boundary condition given by (\ref{fq}). Its solution is solely based on the \textit{index} $\chi$ of $G(z)$ defined as
\begin{displaymath}
\begin{array}{rl}
\chi\equiv &ind_{|z|=1}G(z):=\frac{1}{2\pi i}\int_{|z|=1}d\log\{G(z)\}\\
=&ind_{|z|=1}x(z)-ind_{|z|=1}y(z)\\
&+ind_{|z|=1}[y(z)-\phi_{2}(x(z),y(z))]-ind_{|z|=1}[x(z)-\phi_{1}(x(z),y(z))].
\end{array}
\end{displaymath}

Note that $0<\phi_{k}(x(z),y(z))\leq 1$, $k=1,2$, while $G(z)$, $S(z)$ posses a continuous derivative along $\mathcal{C}$ and consequently, they satisfy the Ho1der condition on $\mathcal{C}$. Since $\mathcal{S}_{k}$, $k=1,2,$ are simple contours and $x(z)$ (resp. $y(z)$) traverses $\mathcal{S}_{1}$ counterclockwise (resp $\mathcal{S}_{2}$ clockwise), we have $ind_{|z|=1}x(z)=1$, $ind_{|z|=1}y(z)=-1$. Moreover,
\begin{displaymath}
\begin{array}{rl}
ind_{|z|=1}[y(z)-\phi_{2}(x(z),y(z))]=&-\frac{1}{2},\\
ind_{|z|=1}[x(z)-\phi_{1}(x(z),y(z))]=&\frac{1}{2}.
\end{array}
\end{displaymath}  
Therefore, $\chi=1$, and the solution of the Riemann boundary value problem is given by,
\begin{equation}
\begin{array}{rl}
\Phi_{1}(z)=&\left\{\begin{array}{ll}
e^{\Gamma(z)}[\Psi(z)+c_{1}z+c_{0}],\,z\in\mathcal{C}^{+},\\
e^{\Gamma^{+}(z)}[\Psi^{+}(z)+c_{1}z+c_{0}],\,z\in\mathcal{C},\\
\end{array}\right.\vspace{2mm}\\
\Phi_{2}(z)=&\left\{\begin{array}{ll}
z^{-1}e^{\Gamma(z)}[\Psi(z)+c_{1}z+c_{0}],\,z\in\mathcal{C}^{-},\\
z^{-1}e^{\Gamma^{-}(z)}[\Psi^{-}(z)+c_{1}z+c_{0}],\,z\in\mathcal{C},\\
\end{array}\right.
\end{array}
\label{sol}
\end{equation}
where $c_{0}$, $c_{1}$, are constants to be determined by using the fact that $\Phi_{1}(0)=\lim_{|z|\to\infty}\Phi_{2}(z)=0$, which in turn yields
\begin{displaymath}
e^{\Gamma(0)}[\Psi(0)+c_{0}]=0,\,\,\,c_{1}=0,
\end{displaymath}
and 
\begin{displaymath}
\begin{array}{rl}
\Gamma(z)=&\frac{1}{2\pi i}\int_{\zeta\in\mathcal{C}}\log\{\zeta^{-1}G(z)\}\frac{d\zeta}{\zeta-z},\,z\in\mathcal{C}\cup\mathcal{C}^{+}\cup\mathcal{C}^{-},\\
\Psi(z)=&\frac{1}{2\pi i}\int_{\zeta\in\mathcal{C}}S(\zeta)e^{-\Gamma^{+}(\zeta)}\frac{d\zeta}{\zeta-z},\,z\in\mathcal{C}\cup\mathcal{C}^{+}\cup\mathcal{C}^{-},
\end{array}
\end{displaymath}
with
\begin{displaymath}
\begin{array}{lr}
\Gamma^{+}(z_{0})=\lim_{z\in\mathcal{C}^{+},z\to z_{0}}\Gamma(z),&\Psi^{+}(z_{0})=\lim_{z\in\mathcal{C}^{+},z\to z_{0}}\Psi(z),\\
\Gamma^{-}(z_{0})=\lim_{z\in\mathcal{C}^{-},z\to z_{0}}\Gamma(z),&\Psi^{-}(z_{0})=\lim_{z\in\mathcal{C}^{-},z\to z_{0}}\Psi(z).
\end{array}
\end{displaymath}

In view of (\ref{sol}), we are able to obtain $\Pi(x(z),0)$, $\Pi(0,y(z))$, as a function of $\Pi(0,0)$. $\Pi(0,0)$ will be obtained by using: $i)$ the normalizing condition $\Pi(1,1)=1$, $ii)$ the solution (\ref{sol}), and $iii)$ the symmetry of the system, which implies that $\Pi(1,0)=\Pi(0,1)$. In particular, setting in (\ref{fun}) $y=1$ and then letting $x\to1$, yields
\begin{equation}
\Pi(1,0)[\bar{s}\alpha-2\bar{s}^{2}a\bar{\alpha}-sl^{-}]+\Pi(0,0)[\bar{s}^{2}\alpha\bar{\alpha}-\bar{s}\alpha]=\lambda+sl^{+}-(s+\bar{s}^{2}\alpha\bar{\alpha}).
\label{s1}
\end{equation}

Using (\ref{sol}), and (\ref{s1}) we obtain $\Pi(0,0)$. Thus, it follows that
$\Pi(x(z),0)$, $z\in\mathcal{C}\cup\mathcal{C}^{+}$ and $\Pi(0,y(z))$, $z\in\mathcal{C}\cup\mathcal{C}^{-}$ are known by (\ref{sol}). Moreover, since
both conformal mappings $x(z)$, $z\in\mathcal{C}\cup\mathcal{C}^{+}$ and $y(z)$, $z\in\mathcal{C}\cup\mathcal{C}^{-}$ unique inverse, $\Pi(x,0)$, $x\in\mathcal{S}_{1}\times\mathcal{S}_{1}^{+}$, and $\Pi(0,y)$, $y\in\mathcal{S}_{2}\times\mathcal{S}_{2}^{+}$ are also known. As a consequence, $\Pi(x,y)$, is now derived by (\ref{fun}).

\subsection{Performance metrics and numerical issues}
Since $\Pi(x,y)$ is obtained in terms of a solution of a Riemann boundary value problem, we are able to derive several performance metrics, such as the expected queue length, and the expected queueing delay in each user buffer. However, since $\Pi(x,y)$ is given in terms of contour integrals, i.e., (\ref{sol}), and the derivation of these integrals require also the numerical derivation of conformal mappings $x(z)$, $y(z)$, several computation problems arise.

In the following, we provide expressions for the expected number of packets at user queues. We focus only on user $U_{1}$. Similar expressions can be derived for user $U_{2}$. Denote by $\Pi_{1}(x, y)$, $\Pi_{2}(x, y)$ the
derivatives of $\Pi(x, y)$ with respect to $x$ and $y$ respectively. Note also that symmetry implies that $\Pi(1,0)=\Pi_{2}(0,1)$, $\Pi_{1}(1,0)=\Pi_{2}(0,1)$, $\Pi_{1}(1,1)=\Pi_{2}(1,1)$. Substituting $y=1$ in (\ref{fun}) and taking the derivative with respect to $x$ at point $x=1$ yields,
\begin{equation}
\begin{array}{rl}
\mathbf{L}_{1}:=\Pi_{1}(1,1)=&\left(\frac{\bar{s}^{2}\alpha\bar{\alpha}-\bar{s}\alpha-sl^{+}}{sl^{-}+\bar{s}^{2}\alpha\bar{\alpha}}\right)\Pi_{1}(1, 0)+W_{1}\Pi(1,0)+W_{0}\Pi(0,0),
\end{array}
\label{per}
\end{equation}
where
\begin{displaymath}
\begin{array}{rl}
W_{1}=&\frac{\lambda(2\bar{s}^{2}\alpha\bar{\alpha}-\bar{s}\alpha-sl^{-})+sl^{+}(\bar{s}^{2}\alpha\bar{\alpha}+sl^{-}-\bar{s}(s+\bar{s}\alpha))}{(sl^{-}+\bar{s}^{2}\alpha\bar{\alpha})^{2}},\vspace{2mm}\\
W_{0}=&\frac{\lambda(\bar{s}\alpha-\bar{s}^{2}\alpha\bar{\alpha})-sl^{+}(sl^{-}+\bar{s}^{2}\alpha\bar{\alpha})-\bar{s}(s+\bar{s}\alpha)}{(sl^{-}+\bar{s}^{2}\alpha\bar{\alpha})^{2}}.
\end{array}
\end{displaymath}

In view of (\ref{per}) we realize that we only need to obtain $\Pi_{1}(1,0)$. Note that the derivation of $\Pi(x,0)$ in terms of a solution of the non-homogeneous Riemann boundary value problem is based on the
properties of the conformal mappings $x:\mathcal{C}^{+}\to\mathcal{S}_{1}^{+}$ and $y:\mathcal{C}^{-}\to\mathcal{S}_{2}^{+}$ derived in (\ref{conf}). Indeed, the properties of the conformal mappings imply that the inverse of these mappings do exist. Let $z=w_{1}(x)$, $z=w_{2}(y)$ the inverse mappings of $x(.)$ and $y(.)$ respectively. Then, the first in (\ref{sol}) yields
\begin{displaymath}
\Pi(x,0)=\Pi(0,0)\times\left[1+\left\{\begin{array}{ll}
e^{\Gamma(w_{1}(x))}[\Psi(w_{1}(x))+c_{1}w_{1}(x)+c_{0}],\,x\in\mathcal{S}_{1}^{+},\\
e^{\Gamma^{+}(w_{1}(x))}[\Psi^{+}(w_{1}(x))+c_{1}w_{1}(x)+c_{0}],\,x\in\mathcal{S}_{1},\\
\end{array}\right.\right],
\end{displaymath}
When $x=1\in\mathcal{S}_{1}$\footnote{If $x=1\in\mathcal{S}_{1}^{-}$, then we need to obtain analytic continuation of the function in the right hand side of (\ref{sol}) following the lines in \cite{nau}. Note that such a case would result in further numerical difficulties.},
\begin{equation}
\begin{array}{rl}
\Pi_{1}(1,0)=&\Pi(0,0)\lim_{x\in S_{1}^{+},x\to 1}\frac{d}{dx}\{e^{\Gamma(w_{1}(x))}[\Psi(w_{1}(x))+c_{1}w_{1}(x)+c_{0}]\}\\
=&\Pi(0,0)\lim_{x\in S_{1}^{+},x \to 1}\left\{w_{1}^{(1)}(x)e^{\Gamma(w_{1}(x))}\left([\Psi(w_{1}(x))+w_{1}(x)c_{1}+c_{0}]\right.\right.\\
&\left.\left.\times\frac{1}{2\pi i}\int_{\tau\in \mathcal{S}_{1}}\frac{\log G(\tau)d\tau}{(\tau-w_{1}(x))^{2}}+\frac{1}{2\pi i}\int_{\tau\in \mathcal{S}_{1}}e^{-\Gamma_{1}^{(+)}(\tau)}\frac{\log S(\tau)d\tau}{(\tau-w_{1}(x))^{2}}\right)+1\right\},
\end{array}
\label{fgg}
\end{equation}
where $w_{1}^{(1)}(x)$ is the first derivative of $w_{1}(x)$ with respect to $x$. Note also that the expressions in (\ref{conf}) contain the function $\lambda(z)$, $z\in\mathcal{C}$, which in turn, corresponds to the solution of the integral equation (\ref{conf1}). Clearly, (\ref{conf1}) can be solved only numerically, whereas the determination of the inverse conformal mappings can also be done numerically. Several techniques, such as trapezoidal rule, and standard iteration procedures has shown rapid convergence based on the values of the parameters.
For a detailed treatment of how we treat numerically (\ref{conf}), (\ref{conf1}), (\ref{fgg}) see \cite{coh1}, Ch. IV.2.

Therefore, although the theory of boundary value problems provides a robust mathematical background to obtain expressions for the pgf of the stationary joint queue length distribution at user buffers, it is not an easy task to obtain numerical results; see \cite{coh1}, part IV. In the next section we provide an efficient approach to derive basic performance metrics without calling for the advanced concepts of theory of boundary value problems.
\section{The Symmetrical System: Mean-value Analysis}\label{sec:bounds}
In the following, we show how we can obtain explicit bounds for the average queueing delay, without calling for the advanced concepts of the theory of boundary value problems. Recall that symmetry implies that $\Pi(1,0)=\Pi_{2}(0,1)$, $\Pi_{1}(1,0)=\Pi_{2}(0,1)$, $\Pi_{1}(1,1)=\Pi_{2}(1,1)$. 

Setting in (\ref{fun}) $y=1$, differentiating with respect to $x$, and letting $x\to 1$, we obtain after some algebra
\begin{equation}
\Pi_{1}(1,1)=\frac{\lambda\bar{\lambda}+sl^{+}\frac{\lambda+sl^{+}}{s+\bar{s}^{2}\alpha\bar{\alpha}}+\Pi_{1}(1,0)[\bar{s}^{2}\alpha\bar{\alpha}-\bar{s}\alpha-sl^{+}]}{s+\bar{s}^{2}\alpha\bar{\alpha}-\lambda-sl^{+}}.
\label{s2}
\end{equation}
To derive (\ref{s2}) we used (\ref{s1}), and the balance condition\footnote{Note that (\ref{s3}) is a conservation of flow relation and equates in steady state the flow of packets into and out of a queue due to the presence of triggering signals.} 
\begin{equation}
sl^{+}\frac{\lambda+sl^{+}}{s+\bar{s}^{2}\alpha\bar{\alpha}}=s\bar{s}l^{+}[1-\Pi(1,0)-\Pi(0,1)+\Pi(0,0)]+sl^{+}[\Pi(1,0)-\Pi(0,0)].
\label{s3}
\end{equation}
Denote $\mathbf{L}=\Pi_{1}(1,1)=\Pi_{2}(1,1)$. Setting $x = y$ in (\ref{fun}), differentiating it with respect to $x$ at $x = 1$ we obtain after lengthy calculations 
\begin{equation}
\begin{array}{rl}
\mathbf{L}=&\frac{4\lambda[1-sl^{-}-\bar{s}^{2}\alpha\bar{a}]+\lambda^{2}-2sl^{-}-2\bar{s}^{2}\alpha\bar{a}+(4\lambda+2)(s+\bar{s}^{2}\alpha\bar{a}-\lambda-sl^{+})+2\Pi_{1}(1,0)[2\bar{s}^{2}\alpha\bar{a}-\bar{s}\alpha+sl^{-}]+(sl^{-})^{2}P(N_{1}>0,N_{2}>0)}{4[s+\bar{s}^{2}\alpha\bar{a}-\lambda-sl^{+}]}.
\end{array}
\label{s4}
\end{equation}
Using (\ref{s3}), (\ref{s4}), and that fact that $0\leq P(N_{1}>0,N_{2}>0)\leq1$ we obtain bounds for the expected queue length at the user buffers. The following theorem summarizes the main result of this section.
\begin{theorem}The following expressions are explicit bounds for the expected queue length at the users' buffers
\begin{equation}
\mathbf{L}_{up}=\mathbf{L}_{low}+\frac{(sl^{-})^{2}[sl^{+}+\bar{s}\alpha-\bar{s}^{2}\alpha\bar{\alpha}]}{2(s+\bar{s}\alpha+sl^{+})(s+\bar{s}^{2}\alpha\bar{\alpha}-\lambda-sl^{+})},\,\,\mathbf{L}_{low}=\mathbf{S},
\end{equation}
where,
\begin{displaymath}
\begin{array}{rl}
\mathbf{S}=&\frac{sl^{+}+\bar{s}\alpha-\bar{s}^{2}\alpha\bar{\alpha}}{2(s+\bar{s}\alpha+sl^{+})(s+\bar{s}^{2}\alpha\bar{\alpha}-\lambda-sl^{+})}\left\{4\lambda[1-sl^{-}-\bar{s}^{2}\alpha\bar{a}]+\lambda^{2}-2sl^{-}-2\bar{s}^{2}\alpha\bar{a}\right.\vspace{2mm}\\&\left.+(4\lambda+2)(s+\bar{s}^{2}\alpha\bar{a}-\lambda-sl^{+})-\frac{2[2\bar{s}^{2}\alpha\bar{a}-\bar{s}\alpha+sl^{-}](\lambda\bar{\lambda}+sl^{+}\frac{\lambda+sl^{+}}{s+\bar{s}^{2}\alpha\bar{\alpha}})}{\bar{s}^{2}\alpha\bar{\alpha}-sl^{+}-\bar{s}\alpha}\right\}.
\end{array}
\end{displaymath}
\end{theorem}
\begin{remark}
Note that when $s=0$ (i.e., classical slotted ALOHA network) and/or $l^{-}=0$ (i.e., only triggering signals are allowed), $\mathbf{L}_{up}=\mathbf{L}_{low}$.
\end{remark}

\section{Numerical Results}\label{sec:num}

In this section we evaluate numerically the theoretical results obtained in the previous sections. We consider the case where the users have the same link characteristics and transmission probabilities to facilitate exposition clarity. 

\subsection{Stable throughput region and stability region}

Here we present the numerical results regarding the stability conditions and the stable throughput region as presented in Section \ref{sec:SCSTR}.
As we mentioned in that section, the stable throughput region is a subset of the stability region. In Figure \ref{fig:Closure}, we present the closure of the stability region and the stable throughput region for the external arrival probabilities for the following three cases for $s_1=s_2=0.1, 0.2, 0.4$. We consider two cases regarding the probabilities $l_i^{+}=0.2, 0.4$. By closure, we mean over all the possible transmission probabilities $\alpha_{i}, i=1,2$.

\begin{figure}[ht!]
\centering
\subfigure{\includegraphics[width=.45\textwidth]{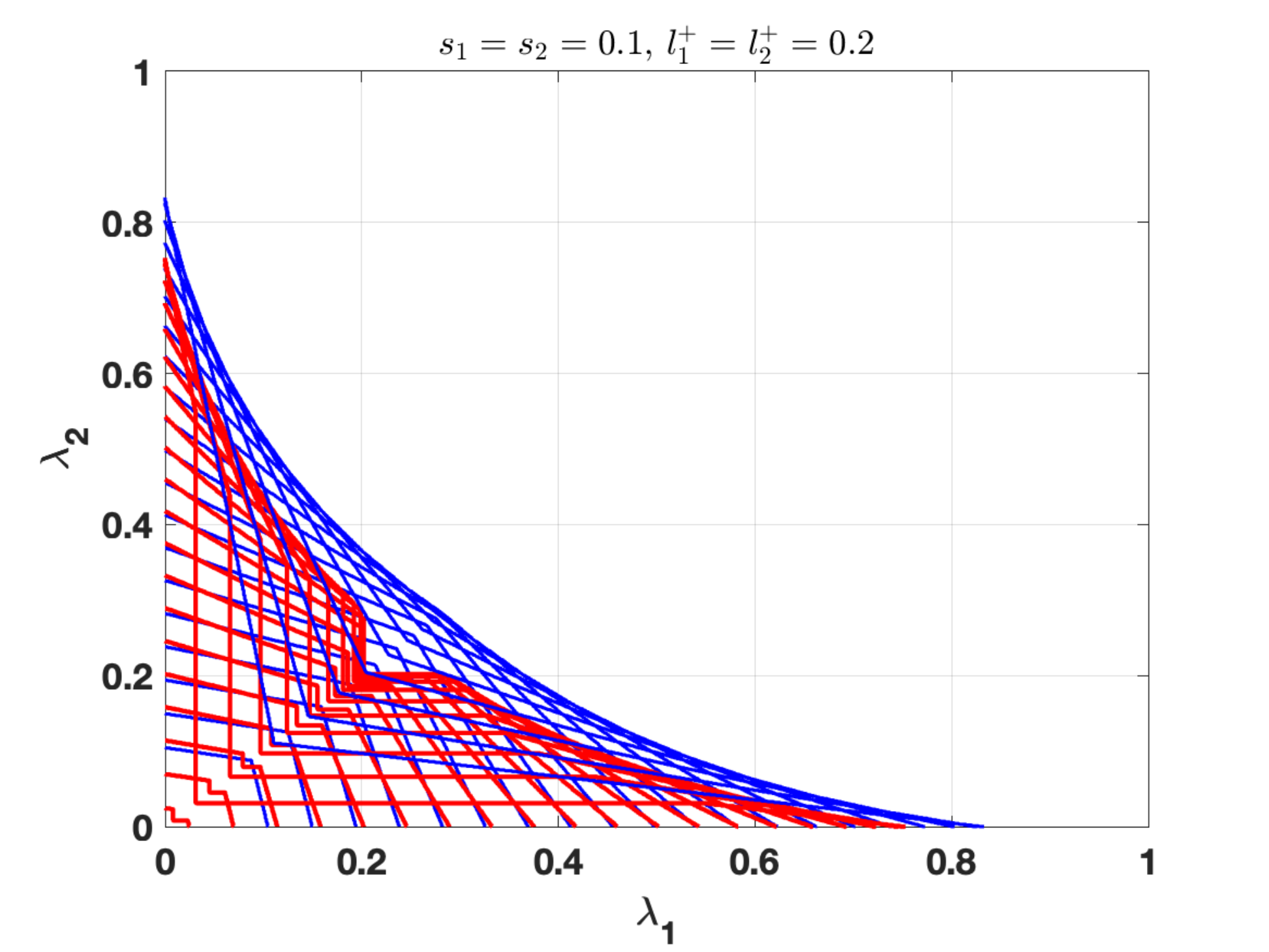}}
\subfigure{\includegraphics[width=.45\textwidth]{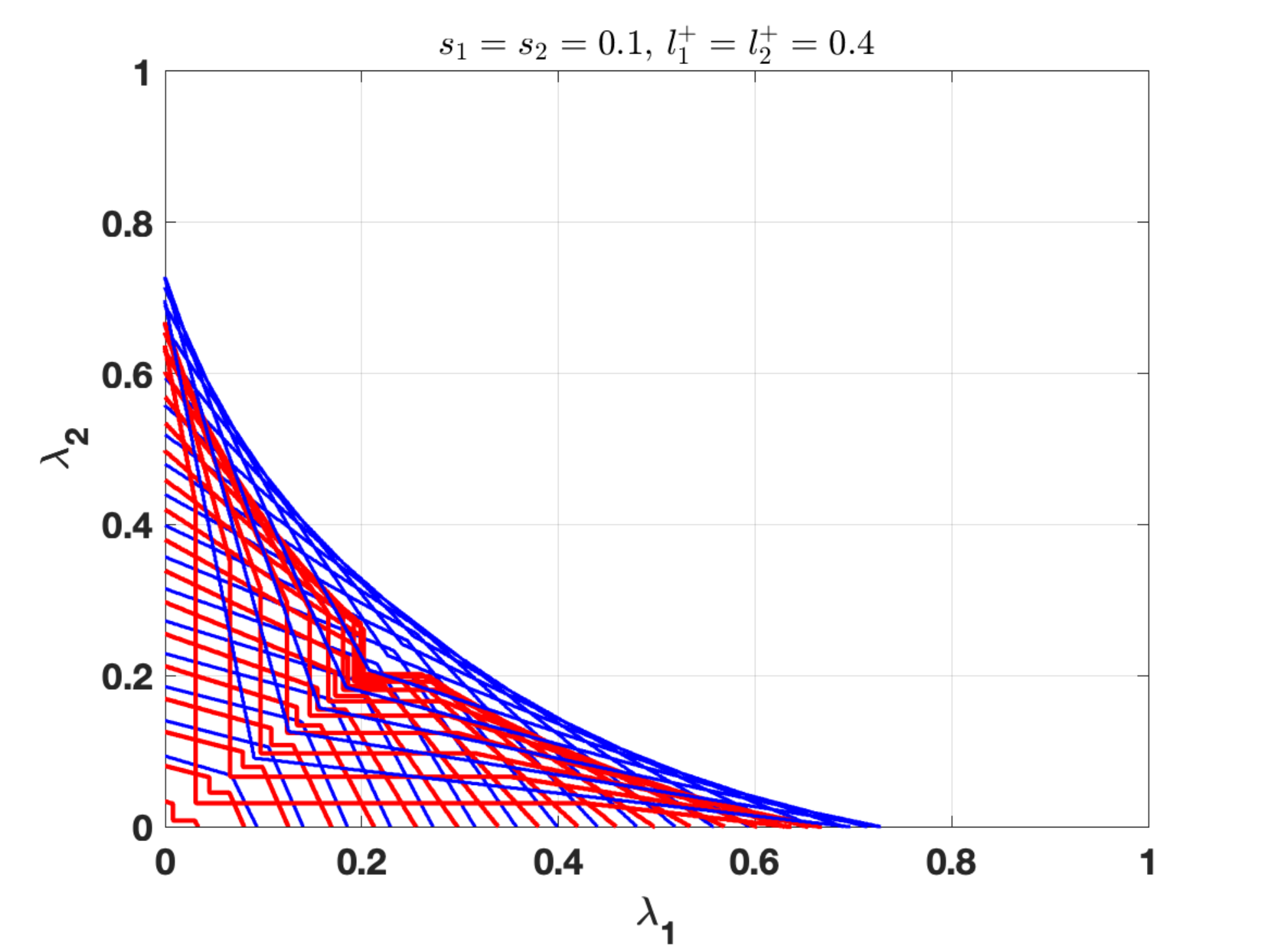}} \\
\subfigure{\includegraphics[width=.45\textwidth]{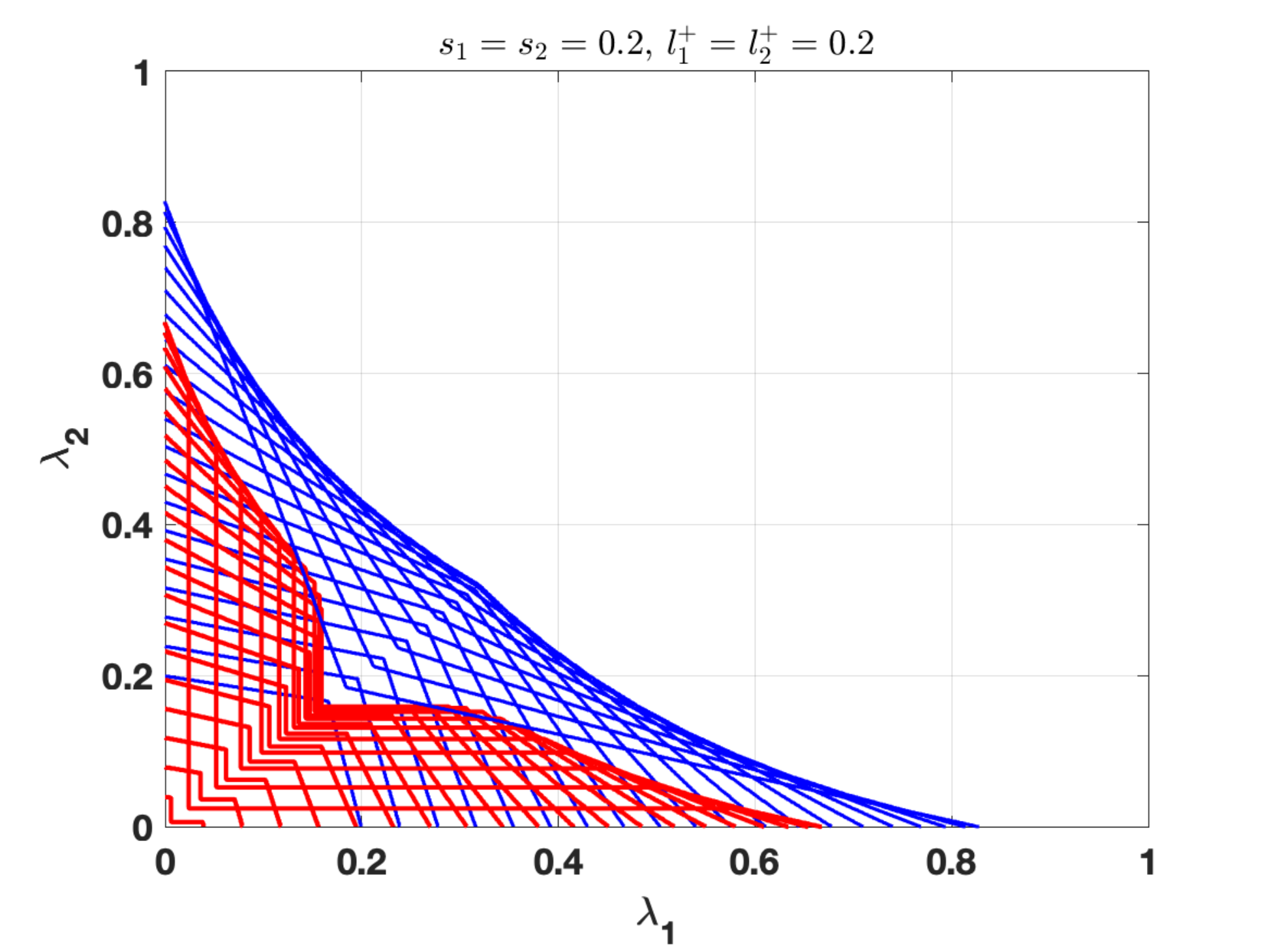}}
\subfigure{\includegraphics[width=.45\textwidth]{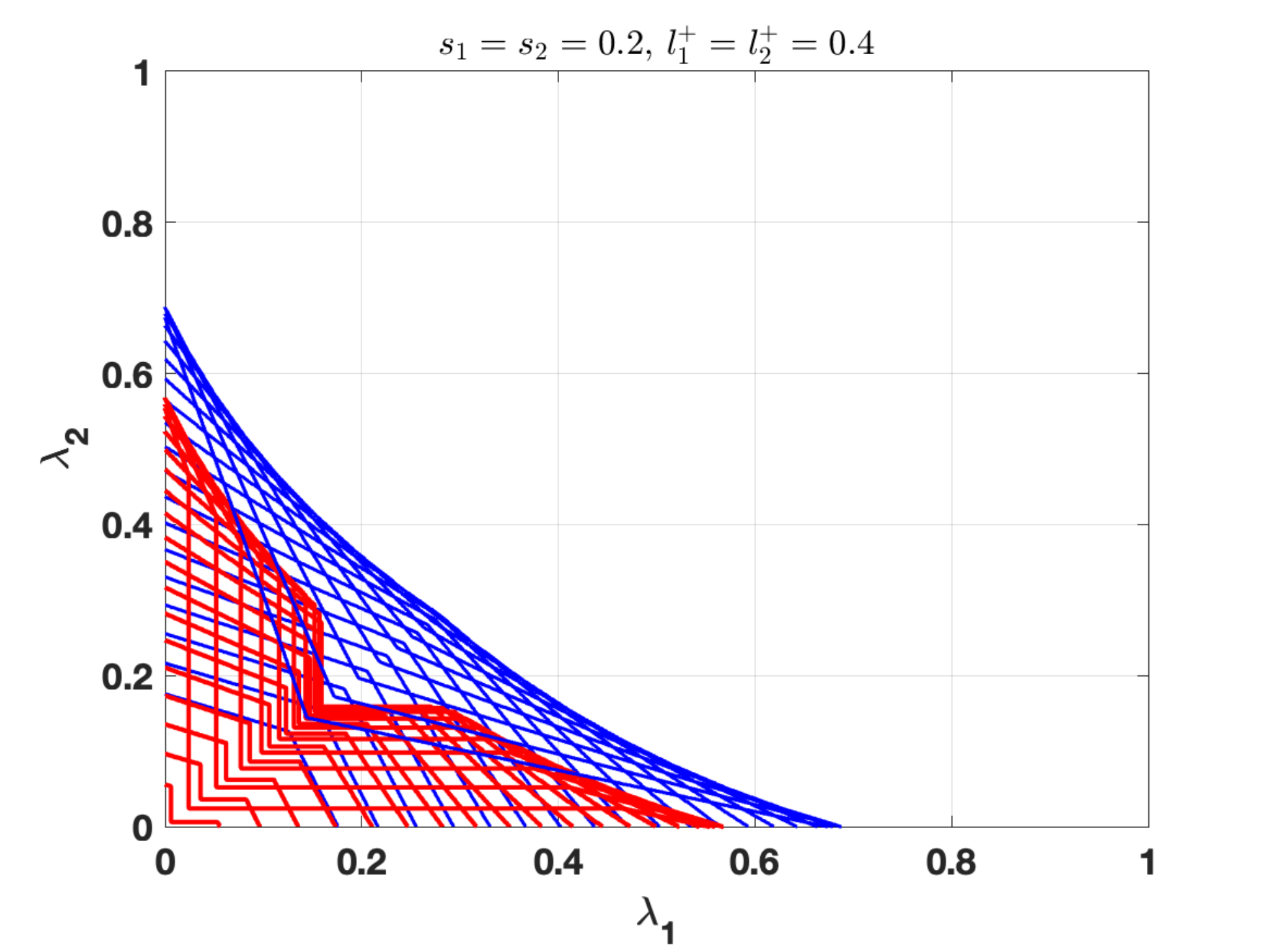}}	\\
\subfigure{\includegraphics[width=.45\textwidth]{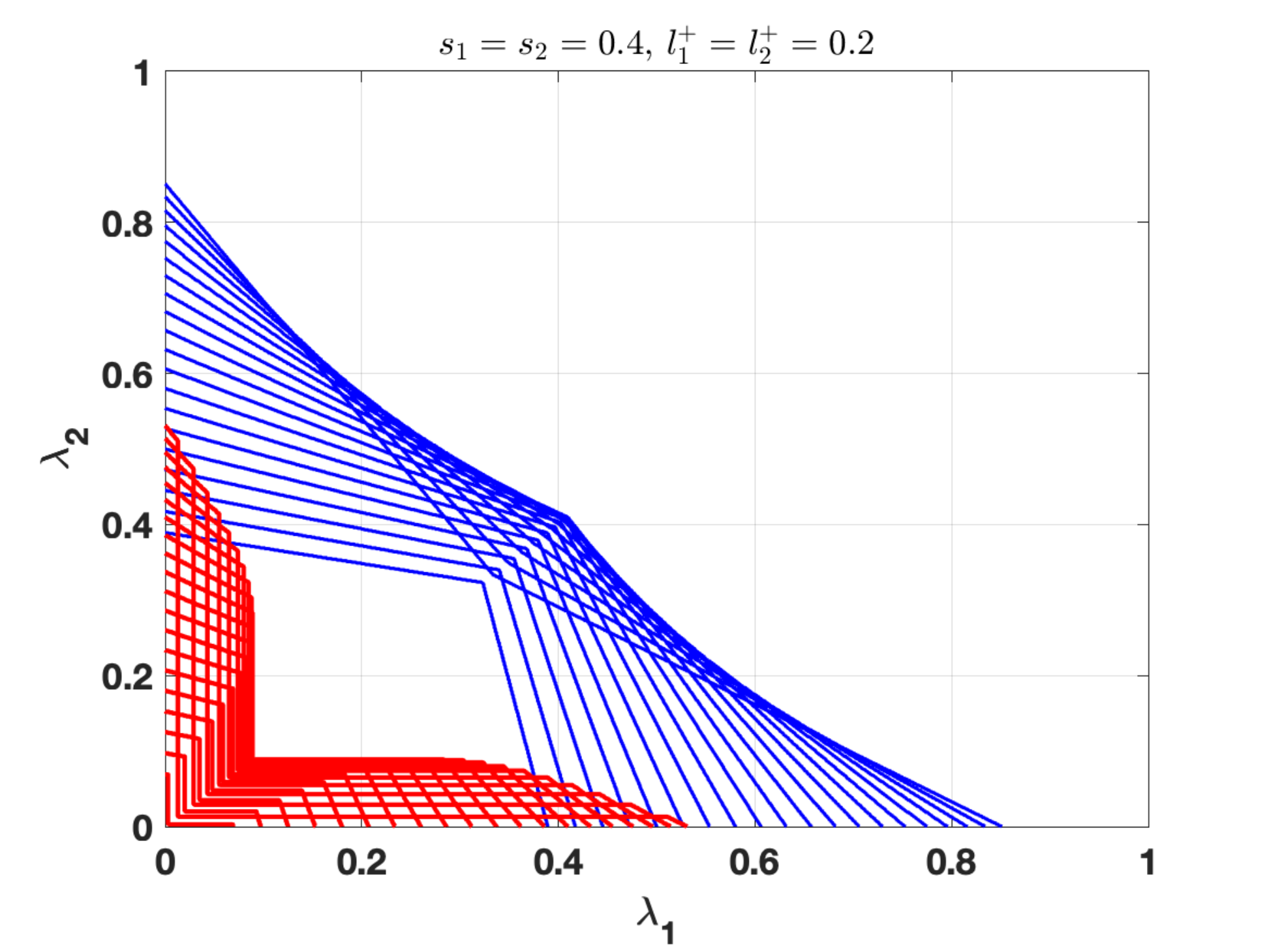}}
\subfigure{\includegraphics[width=.45\textwidth]{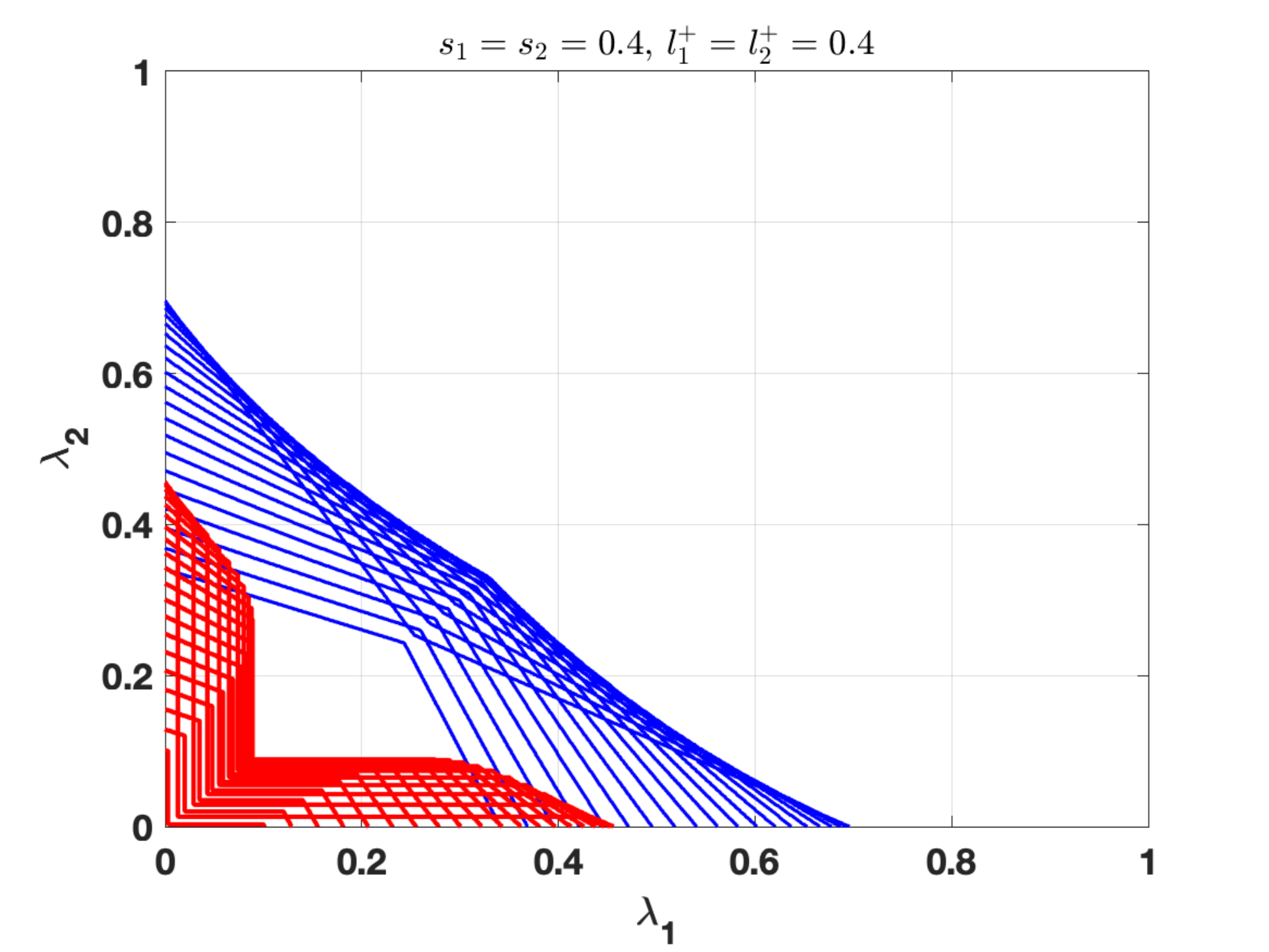}}
\caption{The closure of the stability region and the stable throughput region. The blue lines are for the stability region and the red ones for the stable throughput region.}
\label{fig:Closure}
\end{figure}

As $s_i$ increases, the stability region is becoming slightly broader, this is expected since less packets are transmitted thus we have fewer collisions.
In addition, the stable throughput region is becoming smaller, since the achievable throughput is less, we observe a steep drop in the medium values of the external arrivals, this is because in the expression of an empty queue that we utilized in the dominant systems, the denominator is based on the probability of a packet departure from the queue $m_i$ (including the packet drops that are excluded from the throughput), instead of the probability of service, $\mu_i$.

We observe that for the same value of $s_i$, with the increase of $l_i^{+}$, both regions become smaller, due to internal packet movement which allow for less external traffic arrivals in order to sustain the stability of the system. However, for the stable throughput region, we observe that for the low arrival rate for the one user and the high arrival rate for the other user we got less achievable stable throughout, however, in the medium ones we can achieve a larger region. This is expected due to the smaller value of $l_i^{-}$ which causes the drops of the packets from the system.

\subsection{Performance analysis}
In the following, we focus on how the system parameters affect the expected number of packets.

Figure \ref{fig1} shows the way $l^{+}$ affects the bounds $\mathbf{L}_{low}$, $\mathbf{L}_{up}$, by letting the transmission probability $\alpha$ taking values in $[0.3,0.6]$, and assuming $s=0.2$, $\lambda=0.1$. Clearly, letting $s$ to be small the bounds become tighter, especially when $l^{+}$ increases too. Thus, by enhancing load balancing we can achieve tighter bounds. Moreover, it is easily realized that the increase in transmission probability will decrease the expected number of buffering packets.   

\begin{figure}[ht!]
\centering
\includegraphics[scale=0.6]{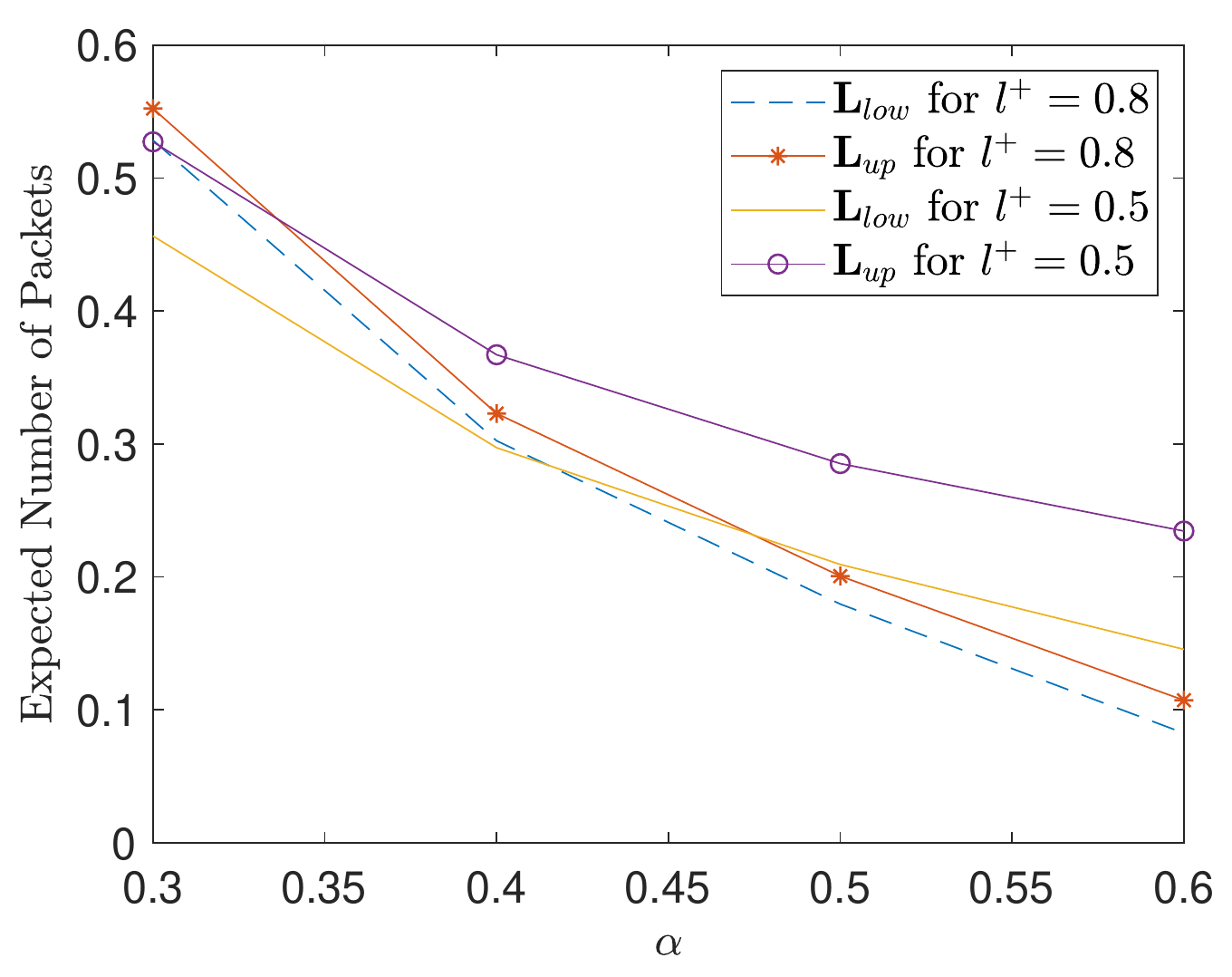}
\caption{Effect of $l^{+}$ on the bounds of the expected number of packets for $s=0.2$, $\lambda=0.1$.}\label{fig1}
\end{figure}

In Figure \ref{fig2} we observe the way $\mathbf{L}_{low}$, $\mathbf{L}_{up}$ are affected for increased values of $l^{+}$. By increasing $l^{+}$ we enhance load balancing. In particular, we can see that as $s$ takes small values, the more we increase $l^{+}$, the better the system behaves, since we can see the decrease on the expected number of buffering packets.
\begin{figure}[ht!]
\centering
\includegraphics[scale=0.6]{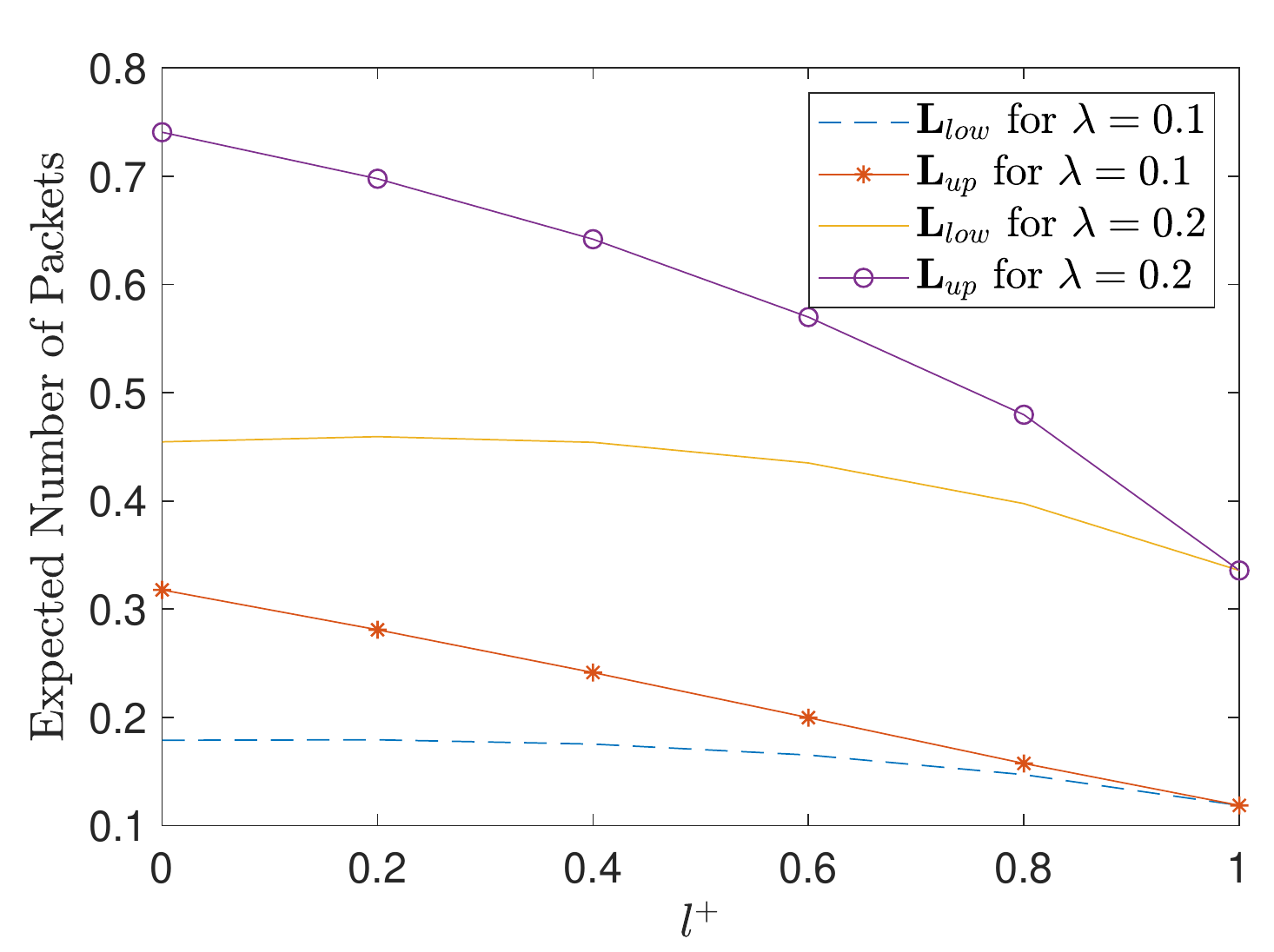}
\caption{Effect of $\lambda$ on the bounds of the expected number of packets for $\alpha=0.6$, $s=0.1$.} \label{fig2}
\end{figure}

Figure \ref{fig3} shows that the increase in $\lambda$ will definitely increase the expected number of buffering packets, but our results shows that in such a case, if $s$ takes large values, the bounds become tighter. This means that the increase in packet arrivals can be balanced by the increase in signal generation, and especially when the probability of a deleting signal (i.e., $l^{-}=0.9$) is large. 

\begin{figure}[ht!]
\centering
\includegraphics[scale=0.6]{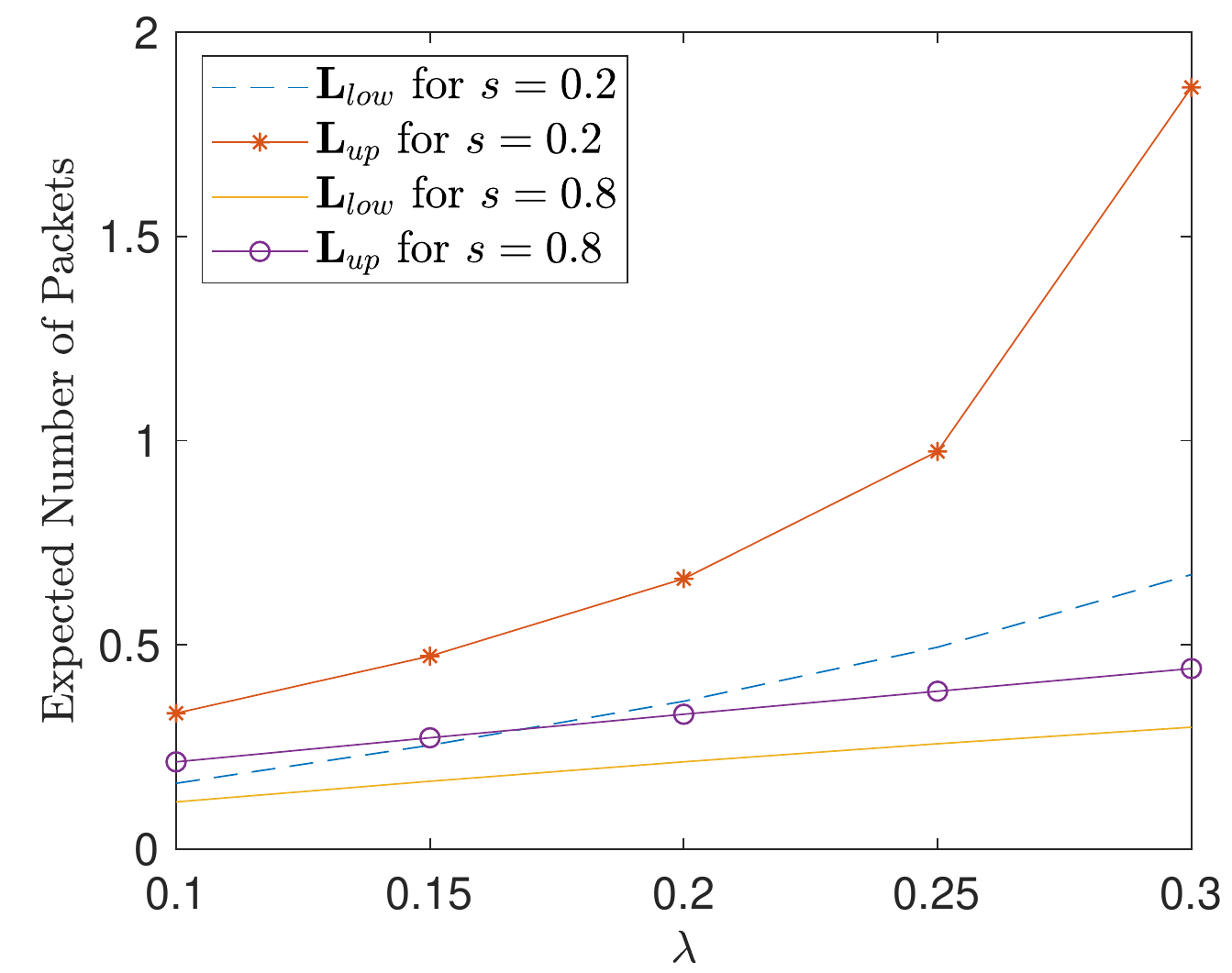}
\caption{Effect of $s$ on the bounds of the expected number of packets for $\alpha=0.6$, $l^{+}=0.1$.} \label{fig3}
\end{figure}
\subsubsection{A comparative study}
In the following, we focus on the impact of signals on the system performance by providing a comparative study between the RAGN and the standard RAN without signals (i.e., $s=0$). For this reason we avoid considering the case of negative signals, and assuming only triggering signals, i.e., $l^{+}=1$. Our aim is to see when load balancing improves the system performance. 
\begin{figure}[ht!]
\centering
\includegraphics[scale=0.52]{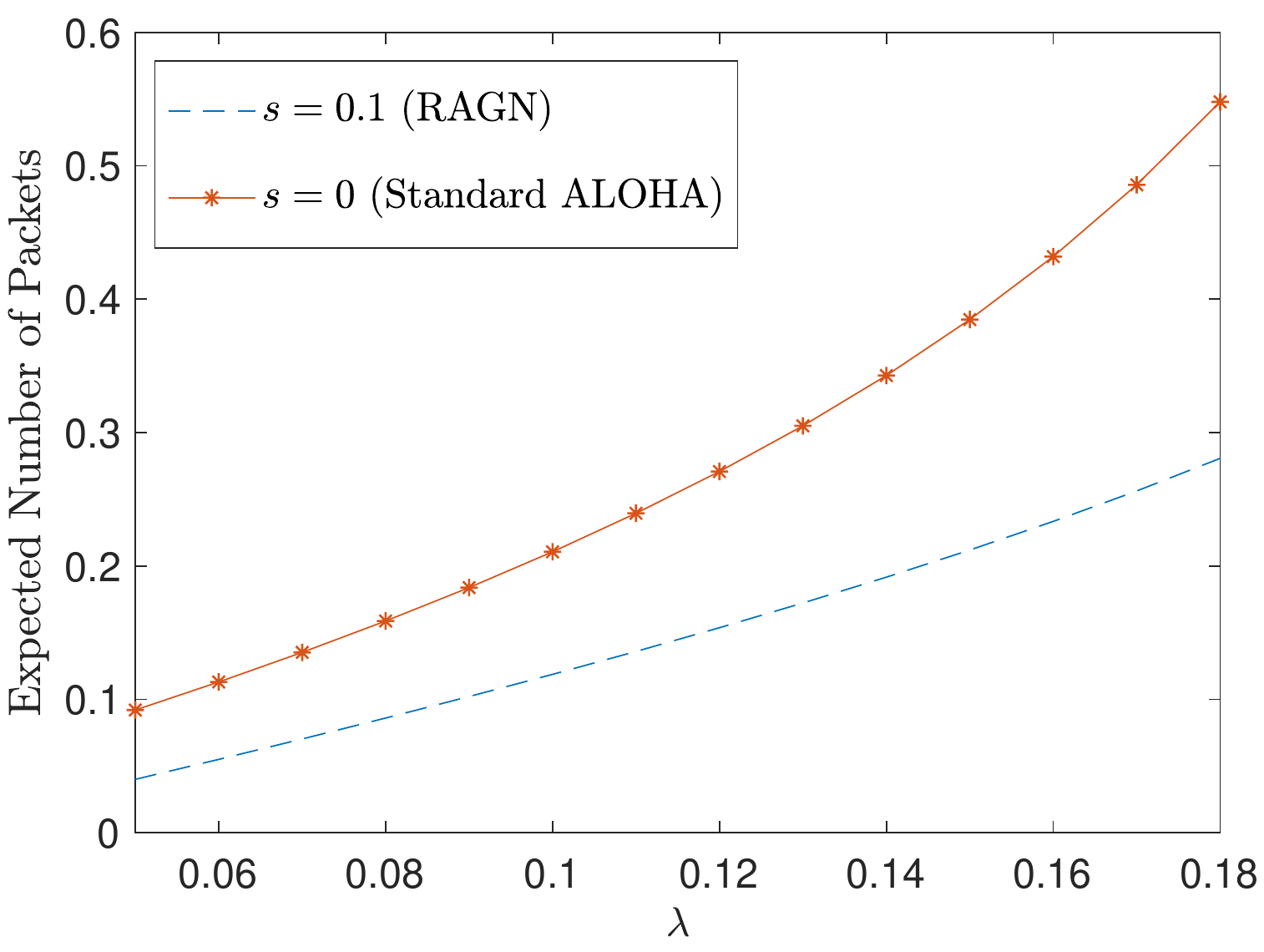}
\caption{Comparison of RAGN and standard ALOHA network for $\alpha=0.6$.} \label{fig4}
\end{figure}

Figure \ref{fig4} shows that RAG network with only triggering signals is superior with respect to the standard ALOHA. In particular, by assuming the the RAG network with only triggering signals (i.e., considering load balancing) with $s=0.1$ (and $l^{+}=1$), we succeeded better performance than the standard ALOHA network (i.e., $s=0$). Clearly, the network performance can be improved further when assuming deleting signals. Thus, the effect of triggering signals will result in better performance and will improve the quality of service in a multiple access network. 

\section{Conclusions and Future Directions}\label{sec:con}
In this paper we introduced for the first time in the related literature the concept of signals in random access networks, by introducing the so-called Random Access G-Network. We considered both negative signals that delete a packet from a user's queue, and triggering signals that cause the instantaneous transfer of packets among user queues. For this interacting network of queues, we obtained both the stability (SR) and the stable throughput (STR) regions using the stochastic dominance approach. We shown that STR is a subset of SR, a result that has not been reported in the related literature. 

Moreover, we provided a compact mathematical analysis and obtain expressions for the pgf of the stationary joint queue length distribution at user queues with the aid of the theory of Riemann boundary value problems. A computationally efficient way to obtain exact bounds for the expected queue length was also given without calling for the advanced concepts of boundary value problems. Numerical results were also obtained and shown useful insights. 

Future extensions of this work include the consideration of more realistic models for the wireless channels such as the erasure and the multi packet reception channel. In addition, security in Internet-of-Things networks can be studied with the consideration of G-networks. Network-level cooperative networks is another interesting direction to be considered.

\section*{Acknowledgements}
This work was supported in part by ELLIIT and the Center for Industrial Information Technology (CENIIT).

\begin{appendices}
\section{An alternative way to derive stability condition}\label{stab}
In the following we derive the stability condition based on Theorem 6.1 in \cite{rw}, pp. 95-96. We start with some necessary notation. Define $\mu_{k}=\frac{d}{dx}\phi_{k}(x,1)|_{x=1}$, $\nu_{k}=\frac{d}{dy}\phi_{k}(1,y)|_{y=1}$, $k=1,2,3$. Then, it is readily seen that
\begin{displaymath}
\begin{array}{rl}
\mu_{3}=&\lambda_{1}+s_{2}l_{2}^{+}+1-(s_{1}+\bar{s}_{1}\bar{s}_{2}\alpha_{1}\bar{\alpha}_{2}),\\
\nu_{3}=&\lambda_{2}+s_{1}l_{1}^{+}+1-(s_{2}+\bar{s}_{1}\bar{s}_{2}\alpha_{2}\bar{\alpha}_{1}),\\
\mu_{1}=&\lambda_{1}+\bar{s}_{1}\bar{\alpha}_{1},\\
\nu_{1}=&\lambda_{2}+s_{1}l_{1}^{+},\\
\mu_{2}=&\lambda_{1}+s_{2}l_{2}^{+},\\
\nu_{2}=&\lambda_{2}+\bar{s}_{2}\bar{\alpha}_{2}.
\end{array}
\end{displaymath}
Denote also $r_{1}=\mu_{1}-1+\nu_{1}\frac{1-\mu_{3}}{1-\nu_{3}}$, $r_{2}=\nu_{2}-1+\mu_{2}\frac{1-\nu_{3}}{1-\mu_{3}}$. We will state only the most interesting part of the ergodicity theorem. For a complete discussion see \cite{rw}. Since under condition $\mu_{3}>1$, $\nu_{3}>1$, the system is unstable (see lemma 6.1, p. 94 in \cite{rw}), we assume hereon that
\begin{displaymath}
\mu_{3}<1\Leftrightarrow \lambda_{1}+s_{2}l_{2}^{+}<(s_{1}+\bar{s}_{1}\bar{s}_{2}\alpha_{1}\bar{\alpha}_{2}).
\end{displaymath}
Then, following Theorem 6.1, pp. 95-96 in \cite{rw} we have the following result:
\begin{theorem}\label{stab1}
\begin{enumerate}
\item If $\nu_{3}<1\Leftrightarrow \lambda_{2}+s_{1}l_{1}^{+}<(s_{2}+\bar{s}_{1}\bar{s}_{2}\alpha_{2}\bar{\alpha}_{1})$, then, $\mathbf{Q}_{n}$ is
\begin{enumerate}
\item positive recurrent iff $r_{1}<0$, $r_{2}<0$, or equivalently
\begin{equation}
\begin{array}{l}
\lambda_{1}+s_{2}l_{2}^{+}\frac{\lambda_{2}+s_{1}l_{1}^{+}}{s_{2}+\bar{s}_{1}\bar{s}_{2}\alpha_{2}\bar{\alpha}_{1}}<(s_{1}+\bar{s}_{1}\alpha_{1})(1-\frac{\lambda_{2}+s_{1}l_{1}^{+}}{s_{2}+\bar{s}_{1}\bar{s}_{2}\alpha_{2}\bar{\alpha}_{1}})+\frac{\lambda_{2}+s_{1}l_{1}^{+}}{s_{2}+\bar{s}_{1}\bar{s}_{2}\alpha_{2}\bar{\alpha}_{1}}(s_{1}+\bar{s}_{1}\bar{s}_{2}\alpha_{1}\bar{\alpha}_{2}),\\
\lambda_{2}+s_{1}l_{1}^{+}\frac{\lambda_{1}+s_{2}l_{2}^{+}}{s_{1}+\bar{s}_{1}\bar{s}_{2}\alpha_{1}\bar{\alpha}_{2}}<(s_{2}+\bar{s}_{2}\alpha_{2})(1-\frac{\lambda_{1}+s_{2}l_{2}^{+}}{s_{1}+\bar{s}_{1}\bar{s}_{2}\alpha_{1}\bar{\alpha}_{2}})+\frac{\lambda_{1}+s_{2}l_{2}^{+}}{s_{1}+\bar{s}_{1}\bar{s}_{2}\alpha_{1}\bar{\alpha}_{2}}(s_{2}+\bar{s}_{1}\bar{s}_{2}\alpha_{2}\bar{\alpha}_{1}).
\end{array}\label{we}
\end{equation}
\item null-recurrent iff $r_{1}\leq0$, $r_{1}=0=r_{2}$, $r_{2}\leq0$.
\item transient iff $r_{1}>0$ or $r_{2}>0$.
\end{enumerate} 
\item If $\nu_{3}\ge1\Leftrightarrow \lambda_{2}+s_{1}l_{1}^{+}\geq(s_{2}+\bar{s}_{1}\bar{s}_{2}\alpha_{2}\bar{\alpha}_{1})$, then, $\mathbf{Q}_{n}$ is
\begin{enumerate}
\item positive recurrent iff $r_{2}<0$, or equivalently
\begin{displaymath}
\begin{array}{l}
\lambda_{2}+s_{1}l_{1}^{+}\frac{\lambda_{1}+s_{2}l_{2}^{+}}{s_{1}+\bar{s}_{1}\bar{s}_{2}\alpha_{1}\bar{\alpha}_{2}}<(s_{2}+\bar{s}_{2}\alpha_{2})(1-\frac{\lambda_{1}+s_{2}l_{2}^{+}}{s_{1}+\bar{s}_{1}\bar{s}_{2}\alpha_{1}\bar{\alpha}_{2}})+\frac{\lambda_{1}+s_{2}l_{2}^{+}}{s_{1}+\bar{s}_{1}\bar{s}_{2}\alpha_{1}\bar{\alpha}_{2}}(s_{2}+\bar{s}_{1}\bar{s}_{2}\alpha_{2}\bar{\alpha}_{1}).
\end{array}
\end{displaymath} 
\item null-recurrent iff $r_{2}=0$.
\item transient iff $r_{2}>0$.
\end{enumerate} 
\end{enumerate}
\end{theorem} 

\end{appendices}

\bibliography{gnet} 
\bibliographystyle{ieeetr}

\end{document}